\documentclass[journal]{IEEEtran}

\usepackage{hyperref}
\hypersetup{hidelinks}
\usepackage{xcolor,soul,framed} 

\usepackage[pdftex]{graphicx}
\DeclareGraphicsExtensions{.pdf,.jpeg,.png}
\usepackage{amsfonts,amssymb}
\usepackage{mathrsfs}
\usepackage{amsthm}
\usepackage[cmex10]{amsmath}
\usepackage{mathtools}
\usepackage{bbm}
\DeclareMathAlphabet{\mathbbb}{U}{bbold}{m}{n}
\usepackage{array}
\usepackage{physics}
\usepackage{cite}
\usepackage{comment}
\newtheorem{theorem}{Theorem}
\newtheorem{lemma}{Lemma}
\newtheorem{definition}{Definition}

\newtheorem{condition}{Condition}
\newtheorem{proposition}{Proposition}

\newtheorem{remark}{Remark}

\let\proof\relax
\let\endproof\relax

\newcommand{\hermconj}  {^{\mathsf{H}}}
\newcommand{\trans}     {^{\mathsf{T}}}
\newcommand{\pha}[1]    {\underline{#1}}
\newcommand{\phaconj}[1]{\overline{\underline{#1}}}
\newcommand{\vect}[1]   {\boldsymbol{#1}}
\newcommand{\phavec}[1] {\underline{\boldsymbol{#1}}}
\newcommand{\phavecconj}[1] {\overline{\underline{\boldsymbol{#1}}}}
\newcommand{\mat}[1]    {\boldsymbol{#1}}
\newcommand{\phamat}[1] {\underline{\boldsymbol{#1}}}
\newcommand{\re}[1]     {\mathscr{R}(#1)}
\newcommand{\im}[1]     {\mathscr{I}(#1)}
\newcommand{\myapprox}{\raisebox{0.5ex}{\texttildelow}}

\begin{document}
\bstctlcite{IEEEexample:BSTcontrol}
    \title{Complex-Frequency Synchronization of Converter-Based Power Systems}
    \author{Xiuqiang~He,~\IEEEmembership{Member,~IEEE,}
          Verena~Häberle,~\IEEEmembership{Student Member,~IEEE,}\\
          and~Florian~Dörfler,~\IEEEmembership{Senior Member,~IEEE}
  \thanks{This work was supported by the European Union’s Horizon 2020 research and innovation program under Grant 883985.}
  \thanks{The authors are with the Automatic Control Laboratory, ETH Zurich, 8092 Zurich, Switzerland. Email:\{xiuqhe,verenhae,dorfler\}@ethz.ch.}}


\maketitle


\begin{abstract}
In this paper, we study phase-amplitude multivariable dynamics in converter-based power systems from a complex-frequency perspective. Complex frequency represents the rate of change of voltage amplitude and phase angle by its real and imaginary parts, respectively. This emerging notion is of significance as it accommodates the multivariable characteristics of power networks where active and reactive power are inherently coupled with both voltage amplitude and phase. We propose the notion of complex-frequency synchronization to study the phase-amplitude multivariable stability issue in a power system with dispatchable virtual oscillator-controlled (dVOC) converters. To achieve this, we separate the system into linear fast dynamics and approximately linear slow dynamics. The linearity property makes it tractable to analyze fast complex-frequency synchronization and slower voltage stabilization. From the perspective of complex frequency and complex-frequency synchronization, we provide novel insights into the equivalence of dVOC and complex-power-frequency droop control, stability analysis methods, and stability criteria. Our study offers a practical solution to address challenging stability issues in converter-based power systems.
\end{abstract}

\begin{IEEEkeywords}
Complex droop control, complex frequency, grid-forming control, multivariable stability analysis, synchronization stability.
\end{IEEEkeywords}

\section{Introduction}

\IEEEPARstart{S}{ynchronization} in AC power systems is a state when the rate of evolution of the phase angle in all generating units is identical \cite{Synchronization-meaning}. This concept is at the core of rotor-angle stability for conventional power systems. Synchronization is also a prerequisite for frequency and voltage stability, as well as smooth and efficient power transmission. However, synchronization is challenging because power systems are inherently \textit{multivariable} and \textit{nonlinear}. Specifically, active power or phase angle ($P$/$\theta$) and reactive power or voltage ($Q$/$V$) are nonlinearly coupled in power systems \cite{Kundur-def}.

In conventional power systems, synchronization is achieved through the interaction of synchronous generators over transmission networks. The dominantly inductive network characteristics result in an approximate $P$/$\theta$ and $Q$/$V$ separation \cite{Kundur-def}. Most analytical studies of synchronization/rotor-angle stability typically assume constant voltage \cite{Bergen-models} and then employ classical swing equation models or the celebrated Kuramoto model. To further address nonlinearity, transient stability is investigated separately from small-disturbance stability, where distinct techniques are applied \cite{Kundur-book}. These techniques are effective in conventional power systems.

Power systems increasingly utilize power converters due to the unprecedented development of renewable energy integration. The loss of synchronism under grid disturbances has occurred in renewable power plants, leading to large-scale generation interruptions \cite{NERC-TR-2017}. Such synchronization stability issues become increasingly challenging due to heterogeneous network characteristics and various converter control strategies. On the network side, $P$/$\theta$ and $Q$/$V$ dynamics become tightly coupled, especially in distribution networks (with low $X/R$ ratios) \cite{De-coupling}, which are increasingly penetrated by distributed energy resources. The sensitivity of load consumption to voltage variations also contributes to this coupling. On the converter-control side, numerous types of control strategies have been developed, e.g., phase-locked loop-based control \cite{He-PLL}, droop control \cite{Huang-droop}, virtual synchronous machine \cite{Ge-VSG}, or matching control \cite{Arghir-matching}, where $P$/$Q$ controls typically occur on the same time scale. The synchronization stability concerning a single converter connected to an infinite grid has been explored for these controls \cite{He-PLL,Huang-droop,Ge-VSG,Arghir-matching}, focusing mainly on the $P$/$\theta$ dynamics of a single converter. Nonetheless, there are few studies on the \textit{multivariable stability analysis concerning both $\theta$ and $V$ dynamics}, especially for \textit{multi-converter interconnected systems}.

Grid-forming (GFM) controls provide preferred solutions for controlling converter-based power systems, particularly microgrids. One of the state-of-the-art GFM controls is \textit{dispatchable virtual oscillator control} (dVOC) \cite{Colombino-dVOC}, which has been shown to achieve almost globally asymptotic stability in terms of synchronization ($\dot \theta$) and voltage ($V$) and meet power setpoints in multi-converter systems \cite{Colombino-dVOC,Gross-dVOC,Subotic-dVOC}. The results have been rigorously established, experimentally validated, and often reproduced \cite{Lu-benchmarking}. Moreover, dVOC reduces to droop control under small voltage amplitude deviations. There is a growing consensus that dVOC is the most performant GFM control \cite{Lu-benchmarking}. However, we emphasize that the stability analysis regarding dVOC in \cite{Colombino-dVOC,Gross-dVOC,Subotic-dVOC} made some restrictive assumptions. Therefore, the theoretical results only apply to a prespecified nominal steady state and a network with a uniform $r/\ell$ ratio. The stability problem with usually drooped steady states in general networks remains open.

A new notion, known as \textit{complex frequency}, has recently been formulated in \cite{Milano-complex-freq}. This notion is both foundational and practical. A complex-frequency quantity, which is the derivative of a \textit{complex angle} \cite{Gu-complex-angle}, represents the rate of change of both voltage logarithm and phase angle. This aligns well with the multivariable coupling characteristics between $P$/$\theta$ and $Q$/$V$ dynamics. With the two-dimensional ``frequency" information, complex frequency has provided practical insights into power system state estimation and control \cite{Zhong-complex-freq,Sanniti-curvature,moutevelis2022taxonomy}. 

In this paper, we study the phase-amplitude multivariable stability issue in a dVOC-controlled power system from the perspective of complex frequency. We define the notion of \textit{complex-frequency synchronization}, reveal that the linear part of dVOC is equivalent to a complex-power-frequency droop control (\textit{complex droop control} for brevity), and show that this complex droop control enables the system to achieve complex-frequency synchronization on a fast time scale. On a relatively slower time scale, the system converges from the synchronous state to a voltage steady state. The slow dynamics are also (approximately) linear when viewed from a ``complex angle" perspective. Thus, we can reformulate the original nonlinear stability problem into two linear subproblems: complex-frequency synchronization and voltage stabilization. We leverage linear system theory to solve these subproblems, effectively handling the case of drooped steady states and nonuniform networks while avoiding the challenge associated with directly treating nonlinear stability. We formulate linear system models, derive quantitative stability conditions, and furnish admittance models and stability criteria. These actionable methods provide a practical solution to the phase-amplitude multivariable stability problem. Moreover, the developed models are extended to encompass more general power systems that include both converters and synchronous generators. Based on this, we discuss how to apply our methods to investigate the stability of the fast converter dynamics in the presence of synchronous generators. The assumptions and results of our analysis are comprehensively validated with high-fidelity nonlinear simulations.

\textit{Notation:} The set of complex numbers is denoted by $\mathbb{C}$. A real or complex scalar is denoted by $x$ or $\pha x$, respectively. The complex conjugate of $\pha x$ is denoted by $\phaconj x$. A vector with real or complex entries is denoted by $\vect x$ or $\phavec x$, respectively. A matrix with real or complex entries is denoted by $\mat{A}$ or $\phamat A$, respectively. Given a matrix $\phamat A \in \mathbb{C}^{m\times n}$, $\phamat A \trans$ and $\phamat A \hermconj$ denote its transpose and Hermitian transpose, respectively. For a real scalar $x$, a complex scalar $\pha x$, a real vector $\vect x$, and a complex vector $\phavec x$, we use $\abs{x}$, $\abs{\pha x}$, $\norm{\vect x}$, $\norm{\phavec x}$ to denote the absolute value, the modulus, the Euclidean norm $(\vect x\trans \vect x)^{1/2}$, and the Euclidean norm $(\phavec x\hermconj \phavec x)^{1/2}$, respectively. For a vector $\phavec x \in \mathbb{C}^{n}$, $\mathrm{diag} \left(\phavec x \right)$ denotes the diagonal matrix formed from $\phavec x$. For a complex scalar, vector, or matrix, $\re{\cdot}$ and $\im{\cdot}$ denote the real and imaginary parts, respectively.

\section{System Description and Problem Statement}
\label{sec-syst-desc}


We consider a three-phase balanced and connected power system, as depicted in Fig.~\ref{fig-converter-based-system}, comprising $N$ converter-interfaced resources such as renewable energy generation, energy storage, etc. The loads in the network are represented by $RLC$ branches or their corresponding constant impedance, which is a common assumption typically used in analytical stability studies \cite{Kundur-book}. We note that while the system herein is purely converter-based, our theoretical results can potentially be applied to more general power systems, cf. Section~\ref{sec-hybrid-system}.

\begin{figure}
  \begin{center}
  \includegraphics[width=8cm]{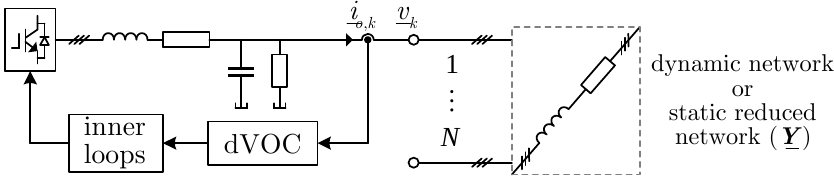}
  \caption{Converter-based power systems, where all converters are controlled by grid-forming dispatchable virtual oscillator control (dVOC).}
  \label{fig-converter-based-system}
  \end{center}
\end{figure}

\textit{1) Power Network:} When ignoring the network dynamics, we obtain a static network representation. A reduced network is further obtained using Kron reduction by eliminating the load nodes. The reduced network is represented by an admittance matrix $\phamat{Y} \in \mathbb{C}^{N \times N}$. For each converter, we define an output current $\pha{i}_{o,k} \in \mathbb{C}$ and a terminal voltage $\pha{v}_k \in \mathbb{C}$. We express $\pha{v}_k$ with the amplitude $v_k$ and the phase angle $\theta _k$ as
\begin{equation}
\label{eq-voltage-expression}
    \pha{v}_k = v_k (\cos\theta_k + j\sin\theta_k) = {v_k}{e^{j{\theta _k}}},
\end{equation}
where $\pha{v}_k$ represents an AC quantity with reference to the stationary frame, and $\theta_k$ represents the phase angle that rotates with time. The power network is described by
\begin{equation}
\label{eq-static-network-equ}
    \phavec{i}_{o} = \phamat{Y}\, \phavec{v},
\end{equation}
where $\phavec{i}_{o} \coloneqq [\pha{i}_{o,1},\cdots,\pha{i}_{o,N}] \trans$ and $\phavec{v} \coloneqq [\pha{v}_{1},\cdots,\pha{v}_{N}] \trans$.
Based on the current-voltage mapping in \eqref{eq-static-network-equ}, the power-flow equations of the network are then given as
\begin{equation}
\label{eq-power-flow-old}
    \pha{s}_k = \textstyle\sum\nolimits_{l = 1}^N \phaconj{y}_{kl} \bigl({v_k}e^{j\theta_k}\bigr) \bigl({v_l} e^{-j\theta_l} \bigr),
\end{equation}
where $\pha{s}_k \coloneqq {p_k} + j {q_k} = \pha{v}_k \phaconj{i}_{o,k}$ denotes the power injection at node $k$, and $\pha{y}_{kl}$ is the $k$th row and $l$th column entry of $\phamat{Y}$.

\textit{2) Grid-Forming dVOC Control:} The converters are controlled by grid-forming dVOC control, whose terminal voltage behavior is given in complex-voltage coordinates as \cite{Colombino-dVOC}
\begin{equation}
\label{eq-dvoc}
    {\pha{\dot v}_k} = \underbrace{j{\omega}_0 {\pha{v}_k} + \eta {e^{j\varphi}}\bigl( \tfrac{p_k^{\star} - jq_k^{\star}}{{v_k^{\star 2}}} {\pha{v}_k} - \overbrace{{\pha{i}}_{o,k}}^{\rm feedback} \bigr)}_{\rm linear} + \underbrace{\eta \alpha \tfrac{{v_k^{\star} - v_k}}{{v_k^{\star}}}{\pha{v}_k}}_{\rm nonlinear},
\end{equation}
where $\omega_0$ denotes the fundamental frequency, $p_k^{\star}$, $q_k^{\star}$, and $v_k^{\star}$ denote the active power, reactive power, and voltage setpoints, respectively, $e^{j\varphi}$ with $\varphi \in [0,\pi/2]$ denotes the rotation operator to adapt to the network impedance characteristics, and $\eta,\alpha > 0$ are control gains. In \eqref{eq-dvoc}, we use complex-valued ordinary differential equations (ODEs) as short-hands for two-dimensional (real-valued) ODEs in $\alpha\beta$ coordinates. In the dVOC control law in \eqref{eq-dvoc}, the first term induces harmonic oscillations of frequency $\omega_0$, the second term synchronizes the relative phase angles to the power setpoints via current feedback, and the third term regulates the voltage amplitude. We notice that the first two terms are linear whereas the third term is nonlinear. Moreover, dVOC has been demonstrated to exhibit a droop-like behavior, making it compatible with conventional power system operation \cite{Lu-benchmarking}.

\textit{3) Multivariable Stability Problem Statement:} The dVOC-based system is formulated by interconnecting the network equation \eqref{eq-static-network-equ} and the node dynamics \eqref{eq-dvoc}, i.e.,
\begin{multline}
\label{eq-dvoc-system}
    \dot {\phavec{v}} = j{\omega}_0 {\phavec{v}} + \eta {e^{j\varphi}}\Bigl[ {\rm diag} \bigl(\bigl \{ ({p_k^{\star} - jq_k^{\star}})/{{v_k^{\star 2}}} \bigr\}_{k=1}^N \bigr) - \phamat{Y} \Bigr] {\phavec{v}} \\
    \quad \, + \eta \alpha {\rm diag} \bigl( \bigl\{ ({v_k^{\star} - v_k})/{v_k^{\star}} \bigr\}_{k=1}^N \bigr) {\phavec{v}},
\end{multline}
where all shunt loads in the network admittance matrix $\phamat{Y}$ can be absorbed into the corresponding setpoints $({p_k^{\star} - jq_k^{\star}})/{{v_k^{\star 2}}}$ such that we consider $\phamat{Y} \in \mathbb{C}^{N \times N}$ a complex-valued and symmetric Laplacian matrix with zero row/column sums.

We aim to study the frequency synchronization and voltage stabilization of the converter-based power system with respect to both $\theta_k$ and $v_k$, or $\pha{v}_k$, through a multivariable analysis that does not rely on the conventional decoupling assumptions.

\begin{remark}
The system in \eqref{eq-dvoc-system} can be related to the well-known Kuramoto oscillator and Stuart-Landau oscillator. The linear part of \eqref{eq-dvoc-system} has been described as the ``linear reformulation of the Kuramoto model" \cite{conteville2012linear}. Furthermore, the particular case of \eqref{eq-dvoc-system} with uniform network $r/\ell$ ratios yields a network of Stuart-Landau oscillators \cite{panteley2020practical}. Since both types of oscillators have been extensively used to model and study oscillation and synchronization behaviors in a wide range of applications including power grids, neuro-physiology, electronic circuits, and seismology \cite{rodrigues2016kuramoto}, the results of this work may also find applications beyond power grids.
\end{remark}

The dVOC-based system has been shown to be almost globally asymptotically stable under certain parametric conditions \cite{Colombino-dVOC,Gross-dVOC,Subotic-dVOC}. The results were achieved by a set of nonlinear analysis approaches. The analysis is theoretically rigorous. However, some restrictive assumptions were required to make the nonlinear problem tractable. As opposed to this, in this work, we take a complex-frequency perspective and provide a novel, more general, and actionable analysis. To do so, we reformulate the system using the notion of complex frequency.

\section{System Reformulation}
\label{sec-reformu}

\subsection{Complex Angle and Complex Frequency}

Consider the complex-voltage expression \eqref{eq-voltage-expression} with nonzero amplitude. We define a complex angle $\pha{\vartheta}_k$ to lead to $\pha{v}_k = e^{\pha{\vartheta}_k}$ \cite{Gu-complex-angle}, which can be seen as a transformation between complex-voltage coordinates and complex-angle coordinates.
\begin{definition}
\label{def-complex-ang}
(Complex angle \cite{Milano-complex-freq,Gu-complex-angle}): Given voltage amplitude $v_k$ (in per unit) and phase angle $\theta_k$ (in rad), we define complex angle ${\pha{\vartheta}_k} \coloneqq {u _k} + j{\theta _k}$, where $u_k \coloneqq \ln{v_k}$ is called voltage logarithm \footnote{$v_k = 0$ is excluded from the definition, which is natural for power system operation. Since $v_k$ is in per-unit and $\ln v_k$ is dimensionless \cite{Milano-complex-freq}, we can still use $\rm rad$ (dimensionless ratio) as the unit of complex angles.}.
\end{definition}

\begin{definition}
\label{def-complex-freq}
(Complex frequency \cite{Milano-complex-freq,Gu-complex-angle}): We define complex frequency $\pha {\varpi}_k \coloneqq \dot {\pha {\vartheta}}_k = {\dot v _k}/v_k + j{\dot \theta _k} = {\varepsilon _k} + j{\omega _k}$ by the time-derivative of ${\pha{\vartheta}_k}$, where $ \varepsilon_k = \dot{v}_k/{v_k}$ is the normalized rate of change of voltage (\textit{rocov}), and $\omega_k$ is angular frequency. The nominal complex frequency is denoted as $\pha{\varpi}_0 \coloneqq j \omega_0$.
\end{definition}

\begin{figure}
  \begin{center}
  \includegraphics{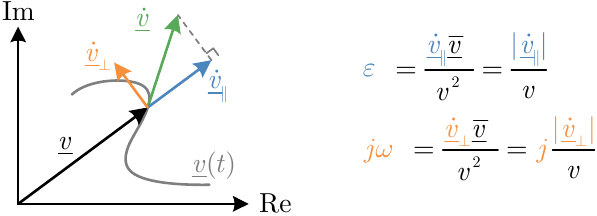}
  \caption{Instantaneous complex frequency $\pha {\varpi} = \varepsilon + j\omega $ is comprised by a radial component $\varepsilon$ and an rotating component $j\omega$.}
  \label{fig-complex-freq-illu}
  \end{center}
\end{figure}

The notion of complex frequency has been recently developed in \cite{Milano-complex-freq} and \cite{Gu-complex-angle}. A geometrical interpretation is provided as follows \cite{Milano-interpretation}. Consider $\pha {\varpi}_k = \dot{\pha{\vartheta}}_k = {\dot{\pha{v}}_k}/{\pha{v}_k} = {\dot{\pha{v}}_k \phaconj{v}_k}/{v_k^2} = {(\dot{\pha{v}}_{\parallel} + \dot{\pha{v}}_{\bot} ) \phaconj{v}_k} / {v_k^2}$, where $\dot{\pha{v}}_k$ is decomposed into a radial component $\dot{\pha{v}}_{\parallel}$ and a rotating component $\dot{\pha{v}}_{\bot}$, which are parallel and perpendicular to $\pha{v}_k$, respectively (see Fig.~\ref{fig-complex-freq-illu}). Since $\dot{\pha{v}}_{\parallel}$ and $\phaconj{v}_k$ have opposite phases, $\dot{\pha{v}}_{\parallel} \phaconj{v}_k$ returns a real number whereas $\dot{\pha{v}}_{\bot} \phaconj{v}_k$ yields an imaginary number. Therefore, the radial and rotating components correspond to the real and imaginary parts of the complex frequency, respectively. In other words, the notion of complex frequency allows for representing two-dimensional frequency information, i.e., the rate of change of voltage in the radial direction and the angular speed in the rotational direction.

We comment that the coordinate $\ln{(v)}$ in Definition~\ref{def-complex-ang} and its differential ${\rm d}v/v$ in Definition~\ref{def-complex-freq} have been used in conventional static and dynamic analysis \cite{Ilic-lnv,Chen-dv-v}, as well as in the representation of converter dynamics \cite{Kogler-norm,moutevelis2022taxonomy}.

\subsection{Complex-Frequency Synchronization}

Next, we introduce the novel concept of complex-frequency synchronization, which, to the best of our knowledge, has not been previously considered or studied in the literature so far.

\begin{definition}
\label{def-complex-freq-sync-def}
(Complex-frequency synchronization): Consider $N$ nodes in a connected network. The voltage trajectories achieve complex-frequency synchronization if all complex frequencies $\dot{\pha{\vartheta}}_k = \dot{\pha{v}}_k / \pha{v}_k$ converge to a common constant complex frequency $\pha{\varpi}_{\rm{sync}}$ as $t \rightarrow \infty$, i.e., $\dot{\pha{\vartheta}}_k \to \pha{\varpi}_{\rm{sync}}, \, t \rightarrow \infty,\, \forall k$.
\end{definition}

Complex-frequency synchronization implies both angular frequency synchronization $\dot \theta _k \to \im{\pha{\varpi}_{\rm{sync}}}$ and rocov synchronization $\dot{{v}}_k / {v}_k \to \re{\pha{\varpi}_{\rm{sync}}}$. Typically, one is interested in the stationary condition for power system operation, where $\dot{{v}}_k \rightarrow 0$, which is also known as \textit{augmented synchronization} \cite{Yang-augmented-sync}. In this work, however, we utilize \textit{nonzero rocov} in conjunction with frequency as a global variable to indicate the imbalance of complex power, encompassing both active and reactive power. This novel concept of complex-frequency synchronization makes multivariable stability analysis more tractable (but it is limited to solely converter dynamics so far). We remark that rocov has been employed as a global indicator for load-generation imbalance in DC grids \cite{Zhao-dc-microgrid}. In this work, we demonstrate that the rocov representation also facilitates stability analysis in AC grids. Moreover, the stabilization of the rocov will also be investigated.

\subsection{Normalized Power-Flow Equations}

In a synchronous state under the classical notion, the power flow in \eqref{eq-power-flow-old} is invariant, i.e., balanced, where the phase angle differences remain invariant while the voltage amplitudes are assumed to be constant. However, for a complex-frequency synchronous state where $\dot{{v}}_k / {v}_k \neq 0$, the power flow in \eqref{eq-power-flow-old} is not invariant due to the constantly-changing voltage amplitudes. To find an invariant representation of load-generation balance in this case, we define the normalized power as $ {\pha{\varsigma}_k} \coloneqq {\rho _k} + j{\sigma _k} \coloneqq \pha{s}_k/{v_k^2} = ({p_k} + j{q_k})/{v_k^2}$, where ${\rho _k} \coloneqq p_k/{v_k^2}$ and ${\sigma _k} \coloneqq q_k/{v_k^2}$ are referred to as normalized active and reactive power, respectively. From \eqref{eq-power-flow-old}, we then derive the \textit{normalized power-flow equations}, their conjugates, and their vector form as follows,
\begin{align}
\label{eq-power-flow-new}
    \pha{\varsigma}_k &= \textstyle\sum\nolimits_{l = 1}^N {\phaconj{{y}}_{kl}} e^{\phaconj{\vartheta}_{l} - \phaconj{\vartheta}_{k}}, \\
\label{eq-power-flow-conj}
    \phaconj{\varsigma}_k &= \textstyle\sum\nolimits_{l = 1}^N {\pha{{y}}_{kl}}e^{\pha{\vartheta}_{l} - \pha{\vartheta}_{k}} = \textstyle\sum\nolimits_{l = 1}^N {\pha{{y}}_{kl}} \frac{\pha{v}_l}{\pha{v}_k}, \\
\label{eq-power-flow-conj-vec}
    \phavecconj{\varsigma} &= [\phaconj{\varsigma}_1,\cdots,\phaconj{\varsigma}_N] \trans = {\rm {diag}} {\bigl( \{1/\pha{{v}}_{k}\}_{k=1}^N \bigr)} \phamat{Y}\, \phavec{v}.
\end{align}

\begin{proposition}
Both the complex-voltage ratio $\pha{v}_l/\pha{v}_k$ and the normalized power $\pha{\varsigma}_k$ remain invariant in a complex-frequency synchronous state.
\end{proposition}
\begin{proof}
From $ \frac{\rm d}{{\rm d}t} \frac{\pha{v}_l}{\pha{v}_k} = \frac{\pha{v}_l}{\pha{v}_k} \left( \frac{\dot {\pha{v}}_l}{\pha{v}_l} - \frac{\dot {\pha{v}}_k}{\pha{v}_k} \right) = \frac{\pha{v}_l}{\pha{v}_k} \bigl({\dot {\pha{\vartheta}}}_{l} - {\dot {\pha{\vartheta}}}_{k} \bigr)$, it follows that $\frac{\rm d}{{\rm d}t} \frac{\pha{v}_l}{\pha{v}_k} = 0$ in the synchronous state where ${\dot {\pha{\vartheta}}}_{l} = {\dot {\pha{\vartheta}}}_{k}$. It then follows from \eqref{eq-power-flow-conj} that $\pha{\varsigma}_k$ also remains invariant.
\end{proof}

The invariance of $\pha{v}_l/\pha{v}_k = v_l/v_k e^{j(\theta_l - \theta_k)}$ implies that both the voltage amplitude ratio and the phase angle difference remain invariant. This property allows voltage amplitudes to constantly change in a synchronous state, which is significant as it drops the constant-voltage assumption. It is of course not relevant in a steady state but useful for analyzing system transients. Furthermore, we note that normalized power coordinates have been utilized in previous studies \cite{Kogler-norm}.

\subsection{Linear Complex Power Flow Equations}
\label{sec-dc-complex-power-flow}

We denote complex-angle differences as $\pha{\vartheta}_{lk} \coloneqq \pha{\vartheta}_{l} - \pha{\vartheta}_{k}$. Assume small ${\pha{\vartheta}}_{lk}$, which implies $v_k \approx v_l$ and $\theta _k \approx \theta _l$. We can then approximate $e^{{{\pha{\vartheta}}_{lk}}}$ as $e^{{{\pha{\vartheta}}_{lk}}} \approx 1 + {\pha{\vartheta}}_{lk}$ by the leading terms of the Maclaurin Series. Substituting this approximation into \eqref{eq-power-flow-conj}, we obtain that $\phaconj{\varsigma}_k \approx \sum\nolimits_{l = 1}^N {{\pha{{y}}_{kl}} \left(1 + {\pha{\vartheta}}_{lk} \right)}$. We thereby define linear complex power flow as follows.
\begin{definition}
(Linear complex power flow) Assume small complex-angle differences ${\pha{\vartheta}}_{lk} \approx 0$, the linear complex power flow $\phaconj{\varsigma}^{\rm dc}_k$ is defined as \footnote{The linear power flow in polar coordinates is also referred to as dc power flow in the literature \cite{Stott-dc-power}. Hence, the superscript ``dc" is employed here.}
\begin{equation}
\label{eq-dc-power-flow}
    \phaconj{\varsigma}^{\rm dc}_k \coloneqq \textstyle\sum\nolimits_{l = 1}^N {{\pha{{y}}_{kl}} \left(1 + {\pha{\vartheta}}_{lk} \right)},
\end{equation}
or formulated in vector form as
\begin{equation}
\label{eq-dc-power-flow-vec}
    \phavecconj{\varsigma}^{\rm dc} = {\rm diag} \bigl(\{1 - \pha{\vartheta}_k\}_{k=1}^N \bigr) \phamat{Y}\mathbbm{1}_N + \phamat{Y}\, \phavec{\vartheta} = \phamat{Y}\, \phavec{\vartheta},
\end{equation}
where $\phavecconj{\varsigma}^{\rm dc} \coloneqq [\phaconj{\varsigma}^{\rm dc}_1,\cdots,\phaconj{\varsigma}^{\rm dc}_N] \trans$, $\phavec{\vartheta} \coloneqq [\pha{\vartheta}_1,\cdots,\pha{\vartheta}_N] \trans$, $\mathbbm{1}_N$ denotes the column vector of all ones, and $\phamat{Y} \mathbbm{1}_N = \mathbbb{0}_N$.
\end{definition}

\begin{remark}
The presented linear complex power-flow equations extend the classical linear (dc) power-flow equations from a complex-angle perspective. The classical ones mainly consider active power \cite{Stott-dc-power}. For a branch with admittance $\pha{y}$ between nodes $k$ and $l$, the linear power flowing from node $k$ to $l$ is $\phaconj{\varsigma}^{\rm dc} = \pha{y} \left({\pha{\vartheta}}_{k} - {\pha{\vartheta}}_{l} \right)$ as in \eqref{eq-dc-power-flow}. The classical linear power flow reads $p^{\rm dc} = ({\theta _k} - {\theta _l})/x$, which holds with two additional assumptions: the network is dominantly inductive, i.e., $\pha{y} = 1/(jx)$, and $v_k \approx v_l \approx 1$. It can be seen that $\phaconj{\varsigma}^{\rm dc}$ degenerates into $p^{\rm dc}$ when applying these two assumptions.
\end{remark}

\begin{remark}
\label{remark-lossless-dc-power}
The linear complex power flow is lossless across the network, as observed by $\mathbbm{1}_N\trans \phavecconj{\varsigma}^{\rm dc} = \mathbbm{1}_N\trans \phamat{Y}\, \phavec{\vartheta} = 0$ from \eqref{eq-dc-power-flow-vec}. This property holds regardless of the network impedance characteristics. We emphasize that the original power flow (without being normalized) is lossy still.
\end{remark}

\subsection{Reformulation of dVOC as Complex Droop Control}

\textit{1) Complex Droop Control:} We separate dVOC into \emph{a linear dVOC core} and the other nonlinear part. The dVOC core is defined by the first two (linear) terms of dVOC in \eqref{eq-dvoc} as
\begin{equation}
\label{eq-dvoc-core}
    {\pha{\dot v}_k} = \pha{\varpi}_0 {\pha{v}_k} + \eta {e^{j\varphi}}\left( \phaconj{\varsigma}_k^{\star} {{\pha{v}}_k} - {{\pha{i}}_{o,k}} \right),
\end{equation}
where $\phaconj{\varsigma}_k^{\star} = ({p_k^{\star} - jq_k^{\star}})/v_k^{\star 2}$ denotes the normalized power setpoint. The dVOC formulation in \cite{Colombino-dVOC} was interpreted from a consensus synchronization perspective. We provide another insightful interpretation from the viewpoint of droop control.

\begin{proposition}
\label{prop-equivalence-droop}
The dVOC core in \eqref{eq-dvoc-core} is equivalent to a complex droop control:
\begin{equation}
\label{eq-augmented-droop-cplx}
    \dot{\pha {\vartheta}}_k = \pha {\varpi}_0 + \eta e^{j\varphi} \left( \phaconj{\varsigma}_k^{\star} - \phaconj{\varsigma}_k\right).
\end{equation}
\end{proposition}
\begin{proof}
The proof is straightforward by expanding \eqref{eq-augmented-droop-cplx} as $\dot{\pha{\vartheta}}_k = \dot{\pha{v}}_k / \pha{v}_k$ and noticing that $\phaconj{\varsigma}_k = ({p_k - jq_k})/v_k^{2} = (\phaconj{v}_k \pha{i}_{o,k})/v_k^{2} = \pha{i}_{o,k}/\pha{v}_k$.
\end{proof}

Proposition~\ref{prop-equivalence-droop} suggests that the dVOC core, similar to a p-f droop control $\dot{\theta}_k = \omega _0 + \eta \left( p _k^{\star} - p _k\right)$ allowing frequency synchronization, is able to achieve complex-frequency synchronization, where the complex frequency $\dot{\pha{\vartheta}}_k$ indicates \textit{the balance of the normalized power $\phaconj{\varsigma}_k$.} This equivalence also suggests the compatibility of dVOC with respect to most existing GFM controllers, including p-f droop control, virtual synchronous machine (VSM), matching control, etc. This compatibility is not surprising since the complex droop control in \eqref{eq-augmented-droop-cplx} augments the classical p-f droop control, which is the core building block of all GFM controls.

\textit{2) Voltage Regulation of Complex Droop Control:} The complex droop control in \eqref{eq-augmented-droop-cplx} aims to regulate the rate of change of voltage (rocov) and frequency, but it cannot regulate voltage amplitude in a direct manner. We augment it with a voltage amplitude regulation term to directly regulate voltage amplitude. One option is to add a proportional control term as in the dVOC in \eqref{eq-dvoc}. Alternatively, we can control voltage logarithm using a dVOC variant similar to the classical one in \eqref{eq-dvoc} as
\begin{equation}
\label{eq-dvoc-log}
    {\pha{\dot v}_k} = \pha{\varpi}_0 {\pha{v}_k} + \eta {e^{j\varphi}}\left( {\phaconj{\varsigma}_k^{\star} {\pha{v}_k} - {{\pha{i}}_{o,k}}} \right) + \eta \alpha \left( {u_k^{\star} - u_k} \right) {\pha{v}_k},
\end{equation}
where $u _k = \ln{v_k}$ and $u _k^{\star} \coloneqq \ln{v_k^{\star}}$. For the dVOC in \eqref{eq-dvoc} and its variant in \eqref{eq-dvoc-log}, the voltage feedback reveals to be highly similar around the operating region where $v_{k} \approx {v_k^{\star}}$. This claim holds due to $1 - v_{k}/{v_k^{\star}} \approx \ln {v_k^{\star}} - \ln v_{k}$ following from the Taylor Series $\ln x = x - 1 + {\rm o}(x-1)$.

In practice, a low-pass filter is typically employed in the voltage feedback to filter out harmonics, which results in
\begin{equation}
\label{eq-dvoc-filter}
\begin{aligned}
    {\pha{\dot v}_k} &= \pha{\varpi}_0 {\pha{v}_k} + \eta {e^{j\varphi}}\left( {\phaconj{\varsigma}_k^{\star} {\pha{v}_k} - {{\pha{i}}_{o,k}}} \right) + \eta \alpha \left( {u_k^{\star} - u_{f,k}} \right) {\pha{v}_k},\\
    \tau \dot{u}_{f,k} &= {u}_k - {u}_{f,k},
\end{aligned}
\end{equation}
where $\tau > 0$ denotes the time constant of the filter and $u_{f,k}$ denotes the filtered version of $u_k$. The dVOC version in \eqref{eq-dvoc-filter} is equivalent to a voltage regulation-enabled complex droop control, which reads in complex-angle coordinates as
\begin{equation}
\label{eq-dvoc-log-filter}
\begin{aligned}
    \dot{\pha {\vartheta}}_k &= \pha {\varpi}_0 + \eta e^{j\varphi} \left( \phaconj{\varsigma}_k^{\star} - \phaconj{\varsigma}_k\right) + \eta \alpha \left( {u_k^{\star} - {u}_{f,k}} \right),\\
    \tau \dot{u}_{f,k} &= {u}_k - {u}_{f,k} = \re{\pha {\vartheta}_k} - {u}_{f,k}.
\end{aligned}
\end{equation}

\subsection{Closed-Loop Model of the Converter-Based Power System}

\textit{1) In Complex-Voltage Coordinates:} Considering complex-voltage coordinates, the closed-loop model of the converter-based power system is established by interconnecting the network equation \eqref{eq-static-network-equ} and the node dynamics \eqref{eq-dvoc-filter} as
\begin{equation}
\label{eq-state-space-filter}
\begin{aligned}
    \dot {\phavec{v}} &= \pha{\varpi}_0 {\phavec{v}} + \eta {e^{j\varphi}}\left( {{\rm diag} \left(\phavecconj{\varsigma}^{\star} \right) - \phamat{Y}} \right) {\phavec{v}} + \eta \alpha {\rm diag} \bigl( \vect{u}^{\star} - \vect{u}_{f} \bigr){\phavec{v}} \\
    \tau \dot{\vect{u}}_f &= \bigl[\ln \abs{\pha{v}_1},\cdots,\ln \abs{\pha{v}_N} \bigr]\trans - \vect{u}_f,
\end{aligned}
\end{equation}
where $\phavecconj{\varsigma}^{\star} \coloneqq [\phaconj{\varsigma}_1^{\star},\cdots,\phaconj{\varsigma}_N^{\star}] \trans$, $\vect{u}^{\star} \coloneqq [u_1^{\star},\cdots,u_N^{\star}] \trans$, and $\vect{u}_f$ is the filtered version of $\vect{u} \coloneqq \bigl[\ln \abs{\pha{v}_1},\cdots,\ln \abs{\pha{v}_N} \bigr]\trans$.

\textit{2) In Complex-Angle Coordinates:} As in \eqref{eq-dvoc-log-filter}, the model in \eqref{eq-state-space-filter} can be transformed into complex-angle coordinates as
\begin{equation}
\label{eq-state-space-log-filter}
\begin{aligned}
    \dot {\phavec{\vartheta}} &= \pha{\varpi}_0 \mathbbm{1}_N + \eta {e^{j\varphi}} \left( {\phavecconj{\varsigma}^{\star} - \phavecconj{\varsigma}} \right) + \eta \alpha \left(\vect{u}^{\star} - \vect{u}_f \right) \\
    \tau \dot{\vect{u}}_f &= \vect{u} - \vect{u}_f = \re{\phavec{\vartheta}} - \vect{u}_f.
\end{aligned}
\end{equation}
where $\vect{u} = \re{\phavec{\vartheta}}$.
Both models in different coordinates are shown in Fig.~\ref{fig-system-models}. The models, as compared to the original one formulated in \eqref{eq-dvoc-system}, employ a variant of dVOC to accommodate complex-angle coordinates. Moreover, a low-pass filter has been incorporated into the voltage feedback.

\begin{figure}
  \begin{center}
  \includegraphics{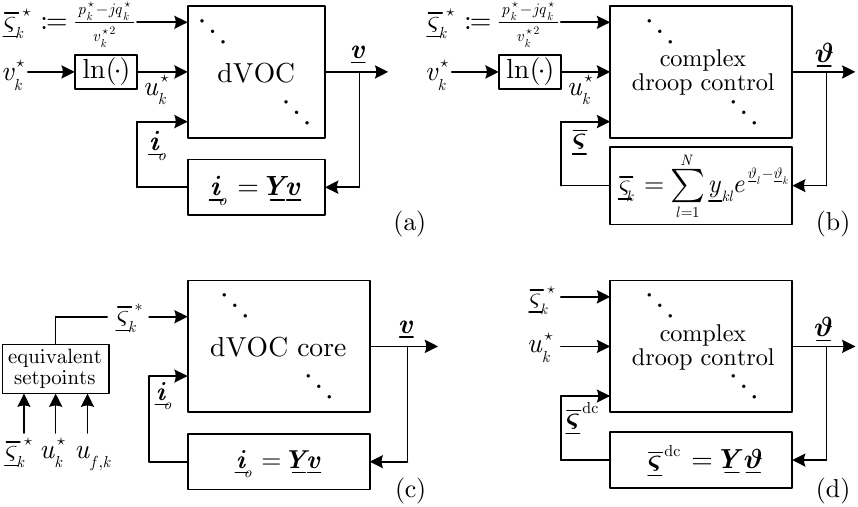}
  \caption{(a) The nonlinear system \eqref{eq-state-space-filter} formulated in complex-voltage coordinates. (b) The equivalent nonlinear system \eqref{eq-state-space-log-filter} formulated in complex-angle coordinates. (c) The linear system \eqref{eq-fast-system} with the dVOC core in \eqref{eq-dvoc-core} represents the fast synchronization dynamics. (d) The linearly approximated system \eqref{eq-slow-system} approximately aims to represent the slower voltage regulation dynamics.}
  \label{fig-system-models}
  \end{center}
\end{figure}

\subsection{Model Separation into Fast and Slow Dynamics}

The system in \eqref{eq-state-space-filter}, similar to dVOC in \eqref{eq-dvoc-filter}, comprises both linear and nonlinear parts. The linear part enables complex-frequency synchronization, while the nonlinear part establishes a voltage steady state. In the original dVOC article \cite{Colombino-dVOC}, these two parts were separated on the time scale by choosing $\alpha$ small, effectively down-tuning the voltage regulation. The low-pass filter in \eqref{eq-state-space-filter} also contributes to the time-scale separation. The time-scale separation between the linear and nonlinear parts allows us to treat their dynamics separately.

\textit{1) Linear System for Fast Dynamics:} On the fast time scale of complex-frequency synchronization dynamics, we neglect the slow dynamics of voltage regulation, i.e., $\vect{u}_f$ in \eqref{eq-state-space-filter} is considered constant. The dynamics of synchronization are then represented by the fast system as
\begin{equation}
\label{eq-fast-system}
    \dot {\phavec{v}} = \pha{\varpi}_0 {\phavec{v}} + \eta {e^{j\varphi}}\left( {\rm diag} \left(\phavecconj{\varsigma}^{\ast} \right) - \phamat{Y} \right) {\phavec{v}},
\end{equation}
where $\phavecconj{\varsigma}^{\ast} \coloneqq \phavecconj{\varsigma}^{\star} + \alpha {e^{-j\varphi}} \left(\vect{u}^{\star} - \vect{u}_f \right)$ denotes the equivalent power reference for the fast system. We note that the fast system \eqref{eq-fast-system} is linear in the \textit{complex-voltage coordinate frame}.

\textit{2) Linearly Approximated System for Slow Dynamics:} On the relatively slower time scale of voltage regulation, we assume that the faster complex-frequency synchronization is already completed. To facilitate a tractable linear analysis of voltage stabilization, we utilize the linear complex power-flow equations derived in Section~\ref{sec-dc-complex-power-flow}. This approximation is reasonable because the system already reaches a synchronous state, where the angle differences $\pha{\vartheta}_{lk}$ are quite small, allowing the linear power-flow approximation to hold well. When replacing $\phavecconj{\varsigma}$ in \eqref{eq-state-space-log-filter} with $\phavecconj{\varsigma}^{\rm dc} = \phamat{Y}\, \phavec{\vartheta}$, a system that approximately represents the slow dynamics can be given as
\begin{equation}
\label{eq-slow-system}
\begin{aligned}
    \dot {\phavec{\vartheta}} &= \pha{\varpi}_0 \mathbbm{1}_N + \eta {e^{j\varphi}} \left( {\phavecconj{\varsigma}^{\star} - \phamat{Y}\, \phavec{\vartheta}} \right) + \eta \alpha \left(\vect{u}^{\star} - \vect{u}_f \right) \\
    \tau \dot{\vect{u}}_f &= \re{\phavec{\vartheta}} - \vect{u}_f.
\end{aligned}
\end{equation}
We refer to this system as a ``linearly approximated system (for slow dynamics)". The linearly approximated system has been formulated in \textit{complex-angle coordinates} to align with the linear complex power-flow equations. In this coordinate frame, it can be seen that the system \eqref{eq-slow-system} is also linear.

\begin{remark}
Both the fast system \eqref{eq-fast-system} and the linearly approximated system \eqref{eq-slow-system} are illustrated in Fig.~\ref{fig-system-models}. We highlight that the choice of coordinates matters because both systems become nonlinear when viewed from a different coordinate frame. Moreover, the linearity of the developed models relies on the assumptions of time-scale separation and linear power flow. Notably, the linearity also benefits from the flexibility in control (particularly, dVOC) of power converters. We further emphasize that both linear models are in the so-called large-signal form, which differs from conventional small-signal models in the incremental form obtained through linearization around specific equilibria. Thanks to this large-signal linearity property, we can employ linear system analysis to yield actionable results on the phase-amplitude multivariable nonlinear stability issue, avoiding the challenges of directly analyzing the nonlinear dynamics of the system \eqref{eq-state-space-filter}.
\end{remark}

\begin{remark}
Time-scale separation is an overarching paradigm in both power systems and power electronics control. Specifically, the entire system under disturbances is supposed to first reach a synchronous state in terms of complex frequency and then a voltage steady state. This does fit empirical evidence observed in practical operation, where the deviation range for voltage is much larger than that for frequency (e.g., \myapprox $\pm 10\%$ for voltage vs. \myapprox $\pm 1.0\%$ for frequency) \cite{European-code}. Since synchronization dynamics are faster than voltage convergence, the two subproblems, complex-frequency synchronization and voltage stabilization, can be analyzed separately. The separated analysis will be accomplished with two mathematically independent systems, i.e., \eqref{eq-fast-system} and \eqref{eq-slow-system}. We remark that this analysis is distinct from the singular perturbation analysis. The standard singular perturbation approach fails when treating synchronization as fast dynamics since the synchronous state of the fast system \eqref{eq-fast-system} is non-isolated \cite[Th. 11.4]{Khalil-nonlinear} (i.e., if $\phavec{v}$ is a synchronous state, then $\pha{k}\,\phavec{v}$ will be too, $\forall \pha{k} \in \mathbb{C} \backslash \{0\}$). 
\end{remark}

The models in \eqref{eq-fast-system} and \eqref{eq-slow-system} exhibit a roughly similar level of abstraction when compared to other widely used models for nonlinear stability analysis in the literature (e.g., the simplified models for droop control \cite{Huang-droop} or virtual synchronous machine \cite{Ge-VSG}). Moreover, for cases, where the dynamics of the inner loops and power networks (faster than synchronization) are of interest, a frequency-domain stability criterion can be employed, as discussed in Section~\ref{sec-s-tomain-analys}. Alternatively, one can also resort to the nested singular perturbation approach \cite{Subotic-dVOC}, which is feasible when treating synchronization as slow dynamics (relatively slow compared to inner-loop dynamics).

\subsection{Model Extension to More General Power Systems}
\label{sec-hybrid-system}

Synchronous generators with inherent swing dynamics cannot achieve complex-frequency synchronization since the inherent interconnection variable is power, not the so-called normalized power. To extend the models in \eqref{eq-fast-system} and \eqref{eq-slow-system} to additionally incorporate the effect of synchronous generators, our focus still remains on the converter dynamics, while we incorporate the voltages generated by synchronous generators as a forced (i.e., exogenous) input to the converter system. To do so, we start by considering the entire network equation as
\begin{equation}
    \begin{bmatrix}
    \phavec{i}_{o} \\ \phavec{i}_{\rm SG}
    \end{bmatrix} =
    \begin{bmatrix}
        \phamat{Y} & \phamat{Y}_{\rm G} \\
        \phamat{Y}_{\rm G}\trans & \phamat{Y}_{\rm SG}
    \end{bmatrix}
    \begin{bmatrix}
        \phavec{v} \\ \phavec{v}_{\rm SG}
    \end{bmatrix},
\end{equation}
where $\phavec{v}_{\rm SG}$ and $\phavec{i}_{\rm SG}$ denote the voltages and currents of synchronous generator nodes, respectively. The admittance matrix is partitioned based on the converter and generator groups. The converter output currents are then expressed as $\phavec{i}_{o} = \phamat{Y}\, \phavec{v} + \phamat{Y}_{\rm G}\phavec{v}_{\rm SG}$. Likewise, the linear power flows are given as
$\phavecconj{\varsigma}^{\rm dc} = \phamat{Y}\, \phavec{\vartheta} + \phamat{Y}_{\rm G} \phavec{\vartheta}_{\rm SG}$,
where $\phavec{\vartheta}_{\rm SG}$ (related to $\phavec{v}_{\rm SG}$) represents the complex angles of synchronous generators.

We augment the system \eqref{eq-fast-system} with a forced input $\phavec{v}_{\rm SG}$, resulting in the following augmented system:
\begin{equation}
\label{eq-fast-system-aug}
    \dot {\phavec{v}} = \pha{\varpi}_0 {\phavec{v}} + \eta {e^{j\varphi}}\left( {\rm diag} \left(\phavecconj{\varsigma}^{\ast} \right) - \phamat{Y} \right) {\phavec{v}} - \eta {e^{j\varphi}} \phamat{Y}_{\rm G}\phavec{v}_{\rm SG}.
\end{equation}
Likewise, we augment the system \eqref{eq-slow-system} with a forced input $\phavec{\vartheta}_{\rm SG}$, which yields that
\begin{equation}
\label{eq-slow-system-aug}
\begin{aligned}
    \dot {\phavec{\vartheta}} &= \pha{\varpi}_0 \mathbbm{1}_N + \eta {e^{j\varphi}} \left( {\phavecconj{\varsigma}^{\star} - \phamat{Y}\, \phavec{\vartheta}} \right) + \eta \alpha \left(\vect{u}^{\star} - \vect{u}_f \right) \\
    & \quad\, - \eta {e^{j\varphi}} \phamat{Y}_{\rm G} \phavec{\vartheta}_{\rm SG}\\
    \tau \dot{\vect{u}}_f &= \re{\phavec{\vartheta}} - \vect{u}_f.
\end{aligned}
\end{equation}

Generally, converter dynamics are notably faster than synchronous generators. As such, we can assume that synchronous generators have fixed voltages and frequencies when investigating the stability of the converter dynamics alone. Specifically, we consider $\phavec{v}_{\rm SG}$ and $\phavec{\vartheta}_{\rm SG}$ as constant exogenous inputs of the linear systems \eqref{eq-fast-system-aug} and \eqref{eq-slow-system-aug}. The time-domain response of linear systems comprises both natural response and forced response. The natural response is of particular importance as it not only reflects the intrinsic behavior of the system but also determines the stability of the forced response under constant inputs. The results of this work are developed for the autonomous systems \eqref{eq-fast-system} and \eqref{eq-slow-system} (i.e., concerning the natural response). Our methods can be applied in power systems that also include synchronous generators to investigate the stability of the fast converter dynamics. However, the stability involving the interaction between converters and generators necessitates further research in the future. The above assumption and discussion will be validated in Section~\ref{sec-case-studies}.

\section{Stability Analysis Methods and Results}
\label{sec-stability-analys}

\subsection{Complex-Frequency Synchronization Analysis}
\label{sec-sync-stab-analys}

We first analyze complex-frequency synchronization stability by focusing on the fast system. The fast system \eqref{eq-fast-system} is a linear autonomous system with system matrix as
\begin{equation}
    \phamat{A} \coloneqq \pha{\varpi}_0 \mat{I}_N + \eta {e^{j\varphi}}\left( {\rm diag} \left(\phavecconj{\varsigma}^{\ast} \right) - \phamat{Y} \right).
\end{equation}

The complex eigenvalues of $\phamat{A}$ are denoted as $\pha{\lambda}_1$, $\pha{\lambda}_2$, $\cdots$, $\pha{\lambda}_N$, where $\re{\pha{\lambda}_1} \ge \re{\pha{\lambda}_2} \ge \cdots \ge \re{\pha{\lambda}_N}$. We term ${\pha{\lambda}_1}$ dominant eigenvalue and assume that $\pha{\lambda}_1$ has algebraic multiplicity one and $\re{\pha{\lambda}_1} > \re{\pha{\lambda}_2}$. This assumption is generic and reflects the actual requirement that power systems normally operate with only one fundamental-frequency component. For the dominant eigenvalue $\pha{\lambda}_1$, we denote by ${\phavec{\psi}_1\trans}$ and ${\phavec{\phi}_1}$ the left and right eigenvectors, respectively, i.e.,
\begin{equation*}
    \phamat{A} \phavec{\phi}_1 = \pha{\lambda}_1 \phavec{\phi}_1, \quad {\phavec{\psi}_1\trans} \phamat{A} = \pha{\lambda}_1 {\phavec{\psi}_1\trans},
\end{equation*}
where simply $\phavec{\psi}_1 = \phavec{\phi}_1$ because $\phamat{A}$ is symmetric. We provide a \textit{spectral stability condition} for complex-frequency synchronization of the fast system \eqref{eq-fast-system}.

\begin{condition}
\label{condition-suff-nece}
(Spectral condition): The entries of ${\phavec{\phi}_1}$ are nonzero, and $\re{\pha{\lambda}_2} < 0$.
\end{condition}

\begin{theorem}
\label{theorem-suff-nece}
Under Condition~\ref{condition-suff-nece}, the system \eqref{eq-fast-system} achieves complex-frequency synchronization at ${\pha{\lambda}_1}$ for almost all initial state ${\phavec{v}_0}$ except for the non-generic initial condition ${\phavec{\psi}_1\trans}{\phavec{v}_0} = 0$ where the voltages $\phavec{v}$ converge to $\mathbbb{0}_N$.
\end{theorem}

The proof is provided in the Appendix. Under Condition~\ref{condition-suff-nece}, the response of the linear system \eqref{eq-fast-system} is dominated by the dominant modal response ${\phavec{\phi}_1}{{\phavec{\psi}_1\trans}{\phavec{v}_0}}{e^{{\pha{\lambda}_1}t}}$. The other modes decay to zero, and consequently, the system synchronizes at $\pha{\lambda}_1$. In Condition~\ref{condition-suff-nece}, if one entry of ${\phavec{\phi}_1}$ is zero (this is not generic), then the dominant modal response of the corresponding node will be zero. If the initial state ${\phavec{v}_0}$ satisfies ${\phavec{\psi}_1\trans}{\phavec{v}_0} = 0$ (this is also not generic), then the overall dominant modal response will be zero. Both cases are not expected by complex-frequency synchronization. Hence, previous work \cite{Colombino-dVOC,Gross-dVOC,Subotic-dVOC}
adopted the terminology of almost global synchronization.

Condition~\ref{condition-suff-nece} relies on the computation of eigenvalues of $\phamat{A}$. Although mathematically intuitive, it may lack physical intuition. To overcome this, we introduce Condition~\ref{condition-suff}, a \textit{parametric stability condition} that allows for gaining physical insights into Condition 1 for complex-frequency synchronization.

\begin{condition}
\label{condition-suff}
(Parametric condition): There exists a maximal phase difference $\bar{\delta} \in [0, \pi/2)$ and a maximal voltage amplitude ratio deviation $\bar{\gamma} \in (0, 1)$ such that $
    \lim\limits_{t \to \infty}\abs{\theta_k - \theta_l} \le \bar{\delta}$ and $ \lim\limits_{t \to \infty}\bigl \lvert \frac{v_k}{v_l} - 1\bigr \rvert \le \bar{\gamma}$
hold for any node pair $(k,l)$ in any synchronous state. Moreover, the power references $\phaconj{\varsigma}_k^{\ast}$, the rotation operator ${e^{j\varphi}}$, and the network admittance matrix $\phamat{Y}$ satisfy
\begin{equation}
\label{eq-suff-formula}
    \max\limits_k \re{{e^{j\varphi}}\phaconj{\varsigma}_k^{\ast}} < \tfrac{{1 + \cos \bar {\delta}}}{2}{\left(1 - \bar{\gamma} \right)^2} {\lambda _2} \bigl(\re{{e^{j\varphi}} \phamat{Y}}\bigr),
\end{equation}
where $\re{{e^{j\varphi}} \phamat{Y}}$ is formed by the real part of the entries of ${e^{j\varphi}} \phamat{Y}$, and it is therefore a Laplacian matrix, and ${\lambda _2}(\cdot)$ denotes the second smallest eigenvalue (a positive real number).
\end{condition}

\begin{theorem}
\label{theorem-suff}
Under Condition~\ref{condition-suff}, the system \eqref{eq-fast-system} achieves complex-frequency synchronization at ${\pha{\lambda}_1}$ for almost all initial state ${\phavec{v}_0}$ except for the non-generic initial condition ${\phavec{\psi}_1\trans}{\phavec{v}_0} = 0$ where the voltages $\phavec{v}$ converge to $\mathbbb{0}_N$.
\end{theorem}

The proof is provided in the Appendix. Condition~\ref{condition-suff} first places a constraint on the phase differences and voltage ratios in the synchronous state. This constraint is intuitively reasonable for power system operation since a healthy synchronous operating point is generally characterized by rather small phase differences and voltage ratios. Condition~\ref{condition-suff} then bounds the power setpoints in terms of the algebraic connectivity ${\lambda _2}(\cdot)$ of the graph corresponding to the real part of the rotated network admittance matrix. \textit{This bound quantifies a few well-known engineering insights}: the network should be sufficiently well connected (as reflected by ${\lambda _2}(\cdot) > 0$) and should not be heavily loaded (cf. the left-hand side of \eqref{eq-suff-formula}), and the rotation angle $\varphi$ should match the impedance angle as closely as possible \cite{De-coupling}. In terms of practical application, Condition~\ref{condition-suff} is actionable in the sense that synchronization stability assessment can be performed in a decentralized fashion with the right-hand side network connectivity information in \eqref{eq-suff-formula}.

Previous dVOC studies \cite{Colombino-dVOC,Gross-dVOC,Subotic-dVOC} have made two restrictive assumptions, i.e., a homogeneous $r/\ell$ ratio in the network and consistent setpoints with power-flow equations. In contrast, \textit{our analysis from the perspective of complex-frequency synchronization drops both assumptions}. Namely, we address a network with arbitrary impedance parameters and consider arbitrary power and voltage setpoints. These two aspects are of significance. First, an actual power network is generally composed of multiple-level sub-networks that are with different impedance characteristics. Second, an actual system almost always operates at a drooped point as a result of load and generation variations.

\subsection{Voltage Stabilization Analysis}
\label{sec-volt-stab-analys}

We show that for the linearly approximated system \eqref{eq-slow-system}, frequency synchronization and voltage stabilization are guaranteed. We rewrite \eqref{eq-slow-system} in real-valued variables as
\begin{equation}
\label{eq-slow-system-real-coeff}
\begin{aligned}
    \dot {\vect{u}} &= \eta \vect{\sigma}^{\prime\star} - \eta \mat{G}^{\prime} \vect{u} + \eta \mat{B}^{\prime} \vect{\theta} + \eta \alpha \left(\vect{u}^{\star} - \vect{u}_f \right) \\
    \dot {\vect{\theta}} &= {\omega}_0 \mathbbm{1}_N + \eta \vect{\rho}^{\prime\star} - \eta \mat{B}^{\prime} \vect{u} - \eta \mat{G}^{\prime} \vect{\theta}\\
    \tau \dot{\vect{u}}_f &= \vect{u} - \vect{u}_f,
\end{aligned}
\end{equation}
where $\vect{u} + j\vect{\theta} = \phavec{\vartheta}$, $\mat{G}^{\prime} + j\mat{B}^{\prime} = {e^{j\varphi}}\phamat{Y}$, and $\vect{\sigma}^{\prime\star} + j\vect{\rho}^{\prime\star} = {e^{j\varphi}} \phavecconj{\varsigma}^{\star}$. In a steady state, $\dot {\vect{u}} = \mathbbb{0}_N$ but $\dot {\vect{\theta}} \neq \mathbbb{0}_N$. To denote the steady-state equilibrium, we define the center-of-angle coordinate as $\theta_0 \coloneqq \mathbbm{1}_N\trans \vect{\theta}/N$. Then ${\dot{\theta}_0} = {\omega _0} + \eta \mathbbm{1}_N\trans \vect{\rho}^{\prime\star}/N$. We denote the angle deviations from $\theta_0$ as $\vect{\delta} \coloneqq \vect{\theta} - \mathbbm{1}_N \theta_0$. The system \eqref{eq-slow-system-real-coeff} is transformed to the center-of-angle coordinate frame as
\begin{equation}
\label{eq-slow-system-coa}
\begin{aligned}
    \dot {\vect{u}} &= \eta \vect{\sigma}^{\prime\star} - \eta \mat{G}^{\prime} \vect{u} + \eta \mat{B}^{\prime} \vect{\delta} + \eta \alpha \left(\vect{u}^{\star} - \vect{u}_f \right) \\
    \dot {\vect{\delta}} &= \eta \left( \vect{\rho}^{\prime\star} - \mathbbm{1}_N \mathbbm{1}_N\trans \vect{\rho}^{\prime\star}/N \right) - \eta \mat{B}^{\prime} \vect{u} - \eta \mat{G}^{\prime} \vect{\delta}\\
    \tau \dot{\vect{u}}_f &= \vect{u} - \vect{u}_f.
\end{aligned}
\end{equation}
We show in Lemma~\ref{lemma-slow-system-equili} in the Appendix that the system \eqref{eq-slow-system-coa} has a unique equilibrium, which is denoted as $[\vect{u}_s\trans, \vect{\delta} _s\trans, \vect{u}_s\trans]\trans$.

\begin{theorem}
\label{theorem-volt-stab-guarant}
    The system \eqref{eq-slow-system-coa} is globally asymptotically stable with respect to the equilibrium $[\vect{u}_s\trans, \vect{\delta} _s\trans, \vect{u}_s\trans]\trans$. The voltages $\vect{u}$ and the frequencies $\dot {\vect{\theta}}$ of the linearly approximated system \eqref{eq-slow-system-real-coeff} converge to $\vect{u}_s$ and ${\omega _0} + \eta \mathbbm{1}_N\trans \im{{e^{j\varphi}} \phavecconj{\varsigma}^{\star}}/N$, respectively.
\end{theorem}

The proof is provided in the Appendix. The statement of ``global stability" only refers to the differential equation \eqref{eq-slow-system-coa} itself. For voltage stabilization of converters, we particularly consider the case where the angle differences are sufficiently small, ensuring the validity of the linearly approximated system \eqref{eq-slow-system-real-coeff}. In this regard, Throrem~\ref{theorem-volt-stab-guarant} provides a guarantee for voltage stabilization of the converter-based system. We interpret the steady-state frequency deviation $\eta \mathbbm{1}_N\trans \im{{e^{j\varphi}} \phavecconj{\varsigma}^{\star}}/N$ as the droop gain multiplied by the average of all the individual rotated active power setpoints. We relate this result to the fact that the linear complex power flow is lossless, cf. Remark~\ref{remark-lossless-dc-power}.

\subsection{Frequency-Domain Stability Criteria}
\label{sec-s-tomain-analys}

We further present frequency-domain admittance modeling and stability criteria. The frequency-domain results are also for global stability, as the developed models are already linear.

\textit{1) Admittance Model of the Fast System:} When deriving the converter equivalent admittance, it is possible to further incorporate the inner-loop dynamics and the $LC$ filter dynamics. The converter admittance can then be derived from $\pha{y}_{{\rm equ},k}(s) \coloneqq -\pha{i}_{o,k}(s)/\pha{v}_k(s)$. For instance, considering solely the dynamics of the dVOC core, we have $\pha{y}_{{\rm equ},k}(s) = (s-j\omega_0){e^{-j\varphi}}/\eta - \phaconj{\varsigma}_k^{\ast}$. Further, when seen from any node $k$, the network-side aggregated admittance, $\pha{y}_{{\rm agg},k}(s)$, can be derived. The admittance is represented as a complex-coefficient transfer function since we assume a three-phase balanced system and an admittance matrix $\bigl[\begin{smallmatrix}
    a(s) & -b(s) \\ b(s) & a(s)
\end{smallmatrix} \bigr]$ is equivalent to $a(s) + jb(s)$ \cite{Xin-SISO-equ}. We remark that $a(s) + jb(s)$ is also known as generalized admittance, see \cite{Xin-SISO-equ} for further details.

\textit{2) Admittance Model of the Linearly Approximated System:} We employ the reduced static network since the system \eqref{eq-slow-system} has been defined with static power flow. For each converter node $k$ in \eqref{eq-slow-system}, we define the rotated power setpoint as $\sigma_{k}^{\prime\star} + j \rho_{k}^{\prime\star} \coloneqq {e^{j\varphi}} \phaconj{\varsigma}_k^{\star}$, cf. the vector form $\vect{\sigma}^{\prime\star} + j\vect{\rho}^{\prime\star} = {e^{j\varphi}} \phavecconj{\varsigma}^{\star}$ in \eqref{eq-slow-system-real-coeff}. Similarly, we define the rotated power flow as $\sigma _k^{\prime {\rm dc}} + j\rho _k^{\prime {\rm dc}} \coloneqq e^{j\varphi}\phaconj{\varsigma}_k^{\rm dc}$. The nodal admittance model is then obtained as
\begin{equation}
\label{eq-adm-equ-dc}
    \underbrace{\frac{1}{\eta} \begin{bmatrix} s + \tfrac{\eta \alpha}{{\tau s + 1}} & 0 \\ 0 & s \end{bmatrix}}_{\mat{Y}_{{\rm equ},k}^{\rm dc}(s)} \begin{bmatrix} {{u_k}} \\ {{\theta _k}} \end{bmatrix} + \begin{bmatrix} {\sigma _k^{\prime {\rm dc}}} \\ {\rho _k^{\prime {\rm dc}}} \end{bmatrix} = \begin{bmatrix} {\sigma_{k}^{\prime\star} + \alpha u_k^{\star}} \\ {\rho_{k}^{\prime\star} + {\omega _0}/{\eta}} \end{bmatrix},
\end{equation}
where $\mat{Y}_{{\rm equ},k}^{\rm dc}(s)$ denotes the converter equivalent admittance and $[{\sigma_{k}^{\prime\star} + \alpha u_k^{\star}},\, {\rho_{k}^{\prime\star} + {\omega _0}/{\eta}}] {\trans}$ denotes the equivalent reference. Further, we can derive the network-side aggregated admittance, denoted by $\mat{Y}_{{\rm agg},k}^{\rm dc}(s)$. These admittance models are represented as real-coefficient $2 \times 2$ transfer function matrices, accounting for the existence of the voltage regulation term.

\textit{3) Admittance-Based Criterion for Synchronization Stability:} Consider $\pha{l}_k(s) \coloneqq \pha{y}_{{\rm agg},k}(s) / \pha{y}_{{\rm equ},k}(s)$ as an admittance ratio. Let $N_1$ be the number of counterclockwise encirclements of the point $(-1+j0)$ for $\pha{l}_k(s)$ when $s$ traverses the Nyquist contour, where $\pha{l}_k(s)$ has $P_1$ unstable poles. Let $Z_1$ be the number of unstable poles of the closed-loop fast system. It follows from the argument principle that $Z_1 = P_1-N_1$. The \textit{stability criterion }is then formulated as $Z_1 \leq 1$, indicating that all closed-loop poles are stable except one accounting for complex-frequency synchronization.

\textit{4) Admittance-Based Criterion for Voltage Stabilization:} Consider $\mat{L}_k(s) \coloneqq \mat{Y}_{{\rm agg},k}^{\rm dc}(s)\mat{Y}_{{\rm equ},k}^{{\rm dc}}(s)^{-1}$ as a return-ratio matrix. Let $N_2$ be the number of counterclockwise encirclements of the origin by $\det\bigl(\mat{I} + \mat{L}_k(s) \bigr)$ when $s$ traverses the Nyquist contour. Note that $\mat{L}_k(s)$ has no unstable poles but has a pole at the origin, cf. \eqref{eq-adm-equ-dc}, and thus, the Nyquist contour should exclude the origin. Moreover, the unity-feedback system of $\mat{L}_k(s)$ is assumed to have no hidden unstable modes. From the generalized Nyquist criterion, the \textit{criterion for voltage stabilization} is then given as $N_2 = 0$.

\section{Case Studies}
\label{sec-case-studies}

We present case studies on a 3-bus converter-based system and a modified IEEE 9-bus system with converters and synchronous generators to illustrate our theoretical results.

\textit{1) Complex-Frequency Synchronization and Voltage Stabilization:} The topology of the 3-bus converter-based system is shown in Fig.~\ref{fig-test-system}, where the line parameters are non-uniform and the converter setpoints are inconsistent. We model the converters using a high-fidelity electromagnetic transient (EMT) model, where both the inner-loop dynamics and the $LC$ filter dynamics are retained and the controllers are implemented in $\alpha\beta$ coordinates \cite{Subotic-dVOC}. We choose control gains of $\eta = 0.04$ and $\alpha = 5$ in per-unit, which provide a frequency droop slope of $\eta = 4\%$ and a voltage droop slope of $1/\alpha = 20\%$. We choose a filter time constant of $\tau = 0.005\, {\rm s}$ (a larger $\tau$ will deteriorate dynamic performance). Moreover, we choose $\varphi = \pi/4$ as a compromise choice to adapt to the heterogeneous $r/\ell$ ratios in the network. We verify that the system is stable for complex-frequency synchronization by checking that Condition~\ref{condition-suff-nece} holds. Furthermore, we compute that ${\lambda _2} \bigl(\re{{e^{j\varphi}} \phamat{Y}}\bigr) \approx 14.2$, highlighting that Condition~\ref{condition-suff} is easily satisfied for arbitrary power setpoints within the capacity range.

\begin{figure}
  \begin{center}
  \includegraphics[width=8cm]{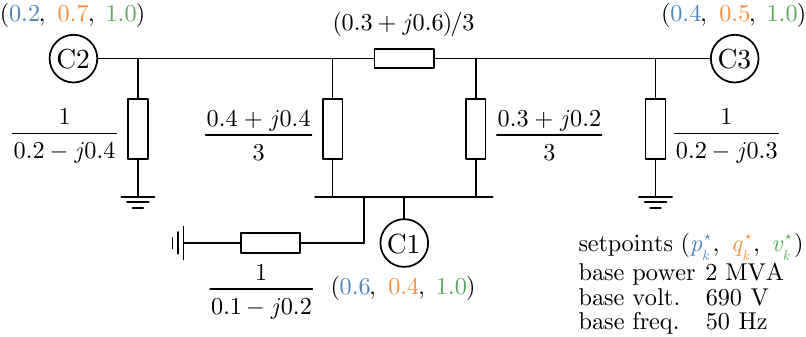}
  \caption{A 3-bus converter-based system, where the network has non-uniform $r/\ell$ ratios and the converters have inconsistent setpoints.}
  \label{fig-test-system}
  \end{center}
\end{figure}

\begin{figure}
  \begin{center}
  \includegraphics{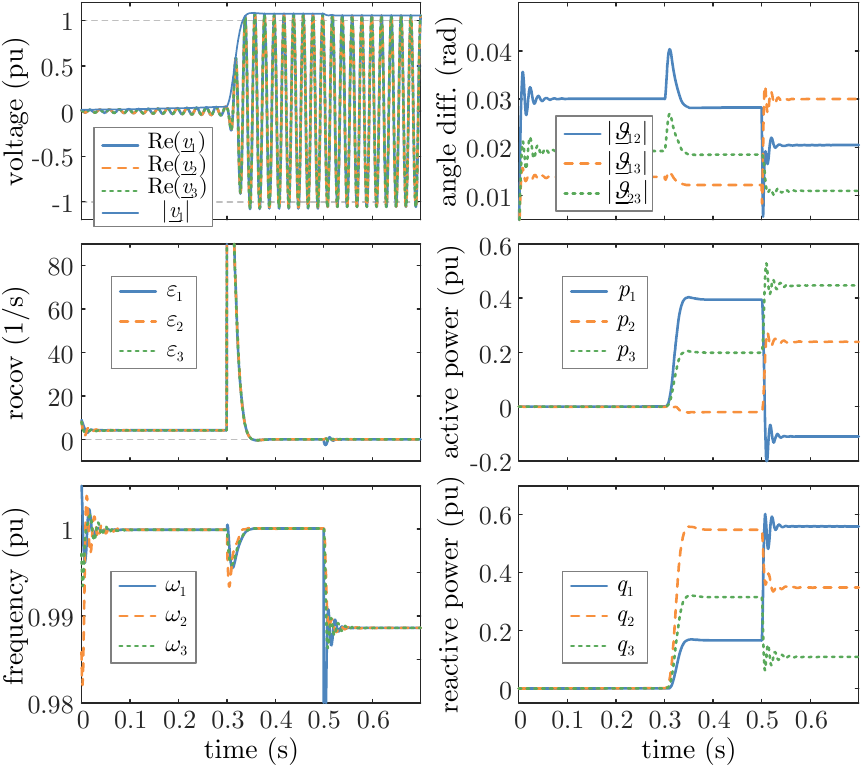}
  \caption{Simulation validation of complex-frequency synchronization, voltage convergence, small complex-angle differences, and power sharing.}
  \label{fig-simulation}
  \end{center}
\end{figure}

We illustrate complex-frequency synchronization and voltage convergence with the simulation results in Fig.~\ref{fig-simulation}. The system initiates from a point near the origin. Before $0.3\ {\rm s}$, voltage regulation is disabled to observe complex-frequency synchronization with the \textit{nonzero} rocov (rate of change of voltage) and frequency waveforms. Notably, the voltages increase exponentially in the synchronous state. At $0.3\ {\rm s}$, voltage regulation is enabled. The voltages are then boosted close to the nominal point, during which complex-frequency synchronization persists. At $0.5\ {\rm s}$, the power setpoint $\phaconj{\varsigma}_1^{\star}$ is changed from $0.6 - j0.4$ to $-0.1 - j0.9$. Subsequently, we observe a rapid process of complex-frequency synchronization, followed by the voltages settling to a new steady state. The active and reactive power sharing occurs in a drooped manner after the transition to steady-state operation. During voltage regulation, the complex-angle differences vary around zero, which supports the assumption of linear complex power flow underlying the previous analysis.

\begin{figure}
  \begin{center}
  \includegraphics{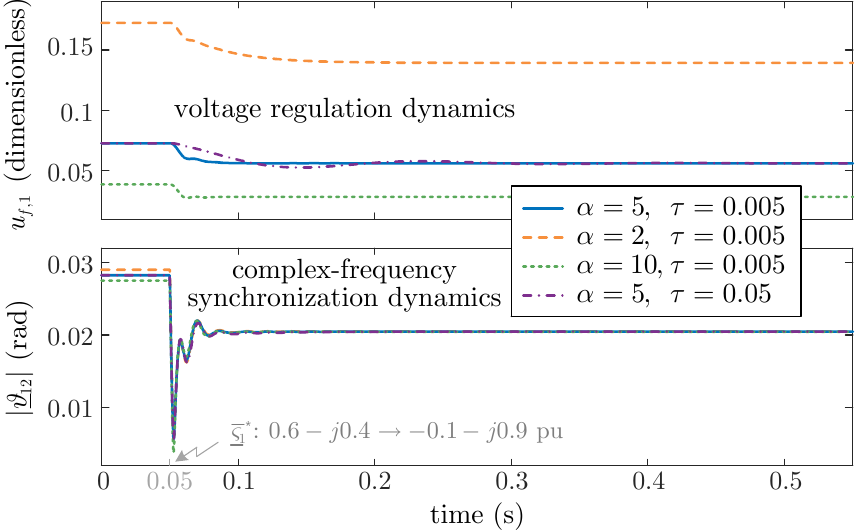}
  \caption{A small $\alpha$ and a large $\tau$ ensure the time-scale separation between complex-frequency dynamics (indicated by $\lvert \pha{\vartheta}_{12} \rvert$) and voltage amplitude regulation (indicated by $u_{f,1}$).}
  \label{fig-interaction}
  \end{center}
\end{figure}

\begin{figure}
  \begin{center}
  \includegraphics{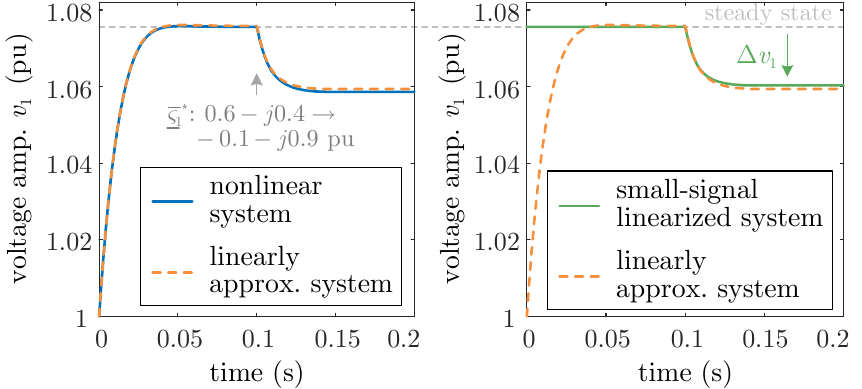}
  \caption{The linearly approximated system \eqref{eq-slow-system} shows a highly accurate large-signal response with respect to the nonlinear system \eqref{eq-state-space-log-filter}. The classical linearized system shows similar accuracy but is limited to a small-signal regime around the specific steady state.}
  \label{fig-linearization}
  \end{center}
\end{figure}

\begin{figure}[t]
  \begin{center}
  \includegraphics{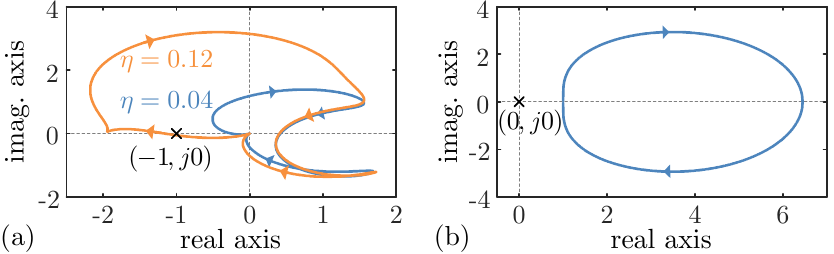}
  \caption{Nyquist diagrams of (a) the fast system and (b) the linearly approximated system.}
  \label{fig-nyquist-plot}
  \end{center}
\end{figure}

\textit{2) Time-Scale Separation Between Synchronization and Voltage Regulation:} The assumption of time-scale separation between complex-frequency synchronization and voltage amplitude regulation is crucial for our theoretical results. As illustrated in Fig.~\ref{fig-interaction}, the time-scale separation depends on the configuration of control parameters. Specifically, a smaller value of $\alpha$ results in a slower response of voltage amplitude (while also altering the steady state of voltage droop control). Additionally, a larger time constant $\tau$ in the voltage feedback filtering also contributes to the time-scale separation (albeit at the cost of longer setting times). When the parameters are chosen appropriately, the time-scale separation is effectively maintained, and both the steady-state operational range and dynamic performance requirements are met. Inappropriate parameter settings, such as $\alpha = 10,\, \tau = 0.005$, may result in time-scale overlap and significant interactions between complex-frequency synchronization and voltage regulation, thus challenging the feasibility of the developed linear methods and stability results.

\textit{3) Accuracy of the Linear Approximation:} The linear power flow approximation is employed to simplify the stability analysis of voltage stabilization. We validate the accuracy of the linear power-flow approximation through the results in Fig.~\ref{fig-linearization}. Both the nonlinear system \eqref{eq-state-space-log-filter} and the linearly approximated one \eqref{eq-slow-system} start from the origin of $\phavec{\vartheta} = \mathbbb{0}_N$. The response of the linearly approximated system very closely matches the nonlinear system response. Notably, the linear system response is characterized by being independent of any steady state, which is facilitated by the adoption of complex-angle coordinates with $\ln(\cdot)$. It is also of interest to examine the response of the classical linearized system around a specific steady state using the usual polar coordinates. To show this, we reformulate the nonlinear system in polar coordinates and then linearize it at the steady-state point, resulting in a small-signal representation (with state variables such as $\Delta v_k$). As observed, the small-signal response of this linearized system exhibits comparable accuracy to the proposed linear system \eqref{eq-slow-system} after the disturbance at $1\, {\rm s}$. However, we emphasize the linearized system response is contingent on the specific steady state, and its accuracy is confined to a small-signal regime. Therefore, the adoption of the conventional linearized system can only yield small-signal stability results.

\textit{4) Application of the Admittance-Based Stability Criteria:} When considering the network dynamics, the synchronization stability is adversely affected \cite{Gross-dVOC}. We test the viability of the admittance-based criterion to assess the destabilizing effect of the network dynamics. The admittance model is established with the dVOC synchronization dynamics and the network dynamics. We compare the synchronization stability under $\eta = 0.04$ and $0.12$. Without considering the network dynamics, both cases are identified to be stable by Conditions~\ref{condition-suff-nece} or \ref{condition-suff}. The system with the presence of the network dynamics, however, becomes unstable under the larger gain $\eta = 0.12$. The Nyquist diagram of $\pha{l}_1(s)$ in Fig.~\ref{fig-nyquist-plot}(a) indicates $N_1 = -1$ in this case, while $\pha{l}_1(s)$ has two unstable poles, so that $Z_1 = 3 > 1$. However, we allow at most one unstable closed-loop pole as synchronous complex frequency. We validate by the Nyquist criterion and eigenvalue calculation that we need $\eta < 0.082$ for stability. We further test the viability of the admittance-based criterion to assess voltage convergence. The Nyquist diagram of $\det\bigl(\mat{I} + \mat{L}_1(s) \bigr)$ is shown in Fig.~\ref{fig-nyquist-plot}(b), where $N_2 = 0$ confirms the voltage stabilization in the case of $\eta = 0.04$. We have validated the admittance model and the criterion under more cases by eigenvalue calculation with the model \eqref{eq-slow-system-to-origin}.

\begin{figure}
  \begin{center}
  \includegraphics[width=8cm]{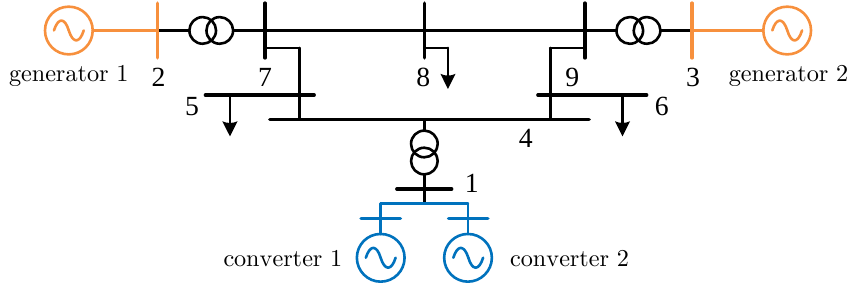}
  \caption{Modified IEEE 9-bus system, where two dVOC-controlled converters replace the synchronous generator at bus 1.}
  \label{fig-hybrid-system}
  \end{center}
\end{figure}

\begin{figure}
  \begin{center}
  \includegraphics{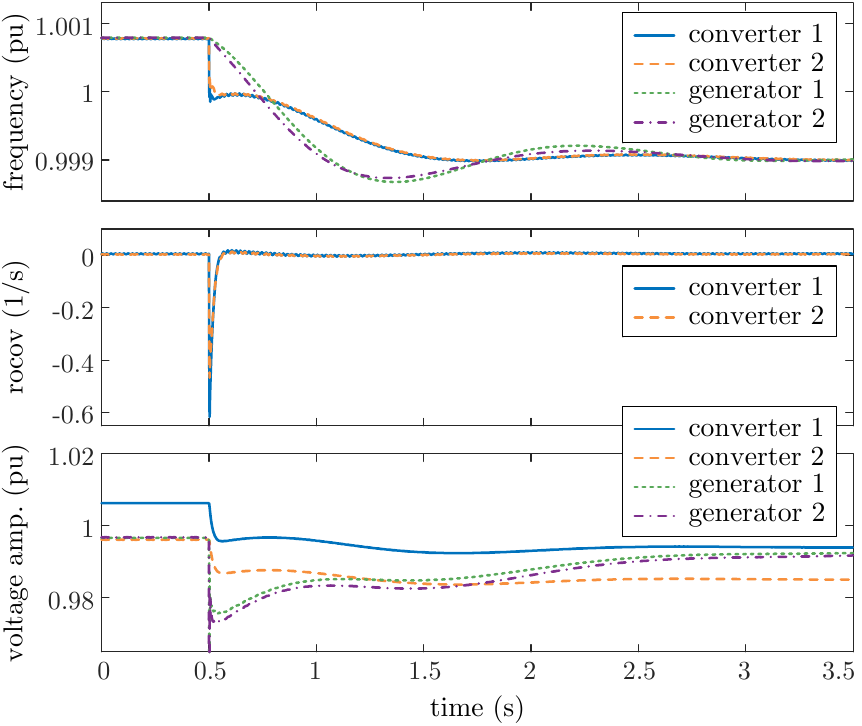}
  \caption{Complex-frequency synchronization and voltage stabilization of the converters are fast after a load increase on bus 6, whereas the frequency and voltage regulation of the generators is relatively slow, validating the modeling assumption that generators provide constant exogenous inputs to converters.}
  \label{fig-hybrid}
  \end{center}
\end{figure}

\textit{5) Response of Converters in Systems With Mixed Generation:} Finally, we observe the behavior of complex-frequency synchronization of converters in a more general power system that also includes synchronous generators. The system configuration is depicted in Fig.~\ref{fig-hybrid-system}, where the generator connected to bus 1 is replaced by two converters. Both converters are governed by the following control parameters: $\eta = 0.03\, {\rm pu}$, $\alpha = 5\, {\rm pu}$, $\tau = 0.005\, {\rm s}$, and $\varphi = \pi/3$. We increase the active power load on bus 6 by $50\, {\rm MW}$. The results in Fig.~\ref{fig-hybrid} indicate the dynamics of the generators are significantly slower than that of the converters due to the mechanical inertia of the rotor. Moreover, the voltage regulation of the generators is impacted by the slowly varying active power. In contrast, the converters achieve (approximate) complex-frequency synchronization rapidly. This suggests that our theoretical results can be utilized to analyze the stability of the fast converter dynamics by assuming that the generator's voltage and frequency are constant. If the fast converter dynamics are stable, then the converters will initially achieve cluster synchronization \cite{menara2020cluster} and subsequently together attempt to achieve full synchronization with the other generators. Regardless of their fast or slow dynamic behaviors, both the converters and generators exhibit droop operation during the steady state, aligning with fundamental regulations of power systems.

\section{Conclusion}
\label{sec-conclusion}

We investigate complex-frequency synchronization and voltage stabilization of dVOC-controlled multi-converter systems. By utilizing complex frequency as a two-dimensional signal to indicate the complex power imbalance, we define the notion of complex-frequency synchronization to investigate the phase-amplitude multivariable dynamics arising from the complex power imbalance. We show that dVOC (i.e., complex droop control) achieves complex-frequency synchronization on the fast time scale and then stabilizes voltage amplitudes on the slower time scale. An insightful understanding of dVOC is provided by revealing its equivalence to complex droop control, thus enabling us to view dVOC as a direct extension of classical p-f droop control. Our theoretical analysis shows that dVOC has superior stability properties. We present both time-domain stability analysis methods and frequency-domain stability criteria. The developed linear methods are practical to deal with the global stability of converter-based power systems such as microgrids and HVDC-connected offshore wind power plants. Furthermore, our methods may also have potential uses in investigating fast converter dynamics in more general power systems that also include synchronous machines.

\appendix

\proof[Proof of Throrem~\ref{theorem-suff-nece}]
Since the dominant eigenvalue ${\pha{\lambda}_1}$ has algebraic multiplicity one, the dominant mode is denoted by $e^{{\pha{\lambda}_1}t}$. The response of $\pha{v}_k$ is composed of a series of modes, where we separate out the dominant modal response as $\pha{v}_k^1$. By spectral decomposition, state equation decoupling, and modal response representation \cite[Sec. 12.2.4]{Kundur-book}, we obtain that ${\phavec{v}^1} \coloneqq [\pha{v}_1^1,\cdots,\pha{v}_N^1]\trans = {\phavec{\phi}_1}{\phavec{\psi}_1\trans}{\phavec{v}_0}{e^{{\pha{\lambda}_1}t}}$, where ${\phavec{v}_0}$ is an initial state. We first consider the case where ${\pha{z}}_0 \coloneqq {\phavec{\psi}_1\trans} {\phavec{v}_0} \neq 0$. We obtain the following limit for any node pair $(k,l)$
\begin{equation}
\label{eq-limit-vkl}
\begin{aligned}
    \lim\limits_{t \to \infty} \frac{{{{\pha{v}}_k}}}{{{{\pha{v}}_l}}} &= \mathop {\lim}\limits_{t \to \infty} \frac{{\pha{v}_k^1 + ( {{{\pha{v}}_k} - \pha{v}_k^1} )}}{{\pha{v}_l^1 + ( {{{\pha{v}}_l} - \pha{v}_l^1} )}} \\
    &= \mathop {\lim}\limits_{t \to \infty} \frac{{{{\pha{\phi}}_{k1}}{{\pha{z}}_0} + ( {{{\pha{v}}_k} - \pha{v}_k^1} ){e^{- {{\pha{\lambda}}_1}t}}}}{{{{\pha{\phi}}_{l1}}{{\pha{z}}_0} + ( {{{\pha{v}}_l} - \pha{v}_l^1} ){e^{- {{\pha{\lambda}}_1}t}}}} = \frac{{{{\pha{\phi}}_{k1}}}}{{{{\pha{\phi}}_{l1}}}},
\end{aligned}
\end{equation}
where ${\pha{\phi}}_{l1} \neq 0$ is the $l$th entry of $\phavec{\phi}_1$, and the third equality holds because $\pha{\lambda}_1$ has a larger real part than the eigenvalues of the other non-dominant modes in $({{{\pha{v}}_k} - \pha{v}_k^1} )$, i.e., $\re{\pha{\lambda}_1} > \re{\pha{\lambda}_2}$. We then obtain for the system \eqref{eq-fast-system} that
\begin{equation}
\label{eq-proof-sync}
\begin{aligned}
    \lim\limits_{t \to \infty} \frac{{\dot {\pha{v}}}_k}{{\pha{v}}_k} &= \pha{\varpi}_0 + \eta {e^{j\varphi}} \Bigl( \phaconj{\varsigma}_k^{\ast} - \lim\limits_{t \to \infty} \textstyle\sum\nolimits_{l = 1}^N {\pha{{y}}_{kl}} \frac{\pha{v}_l}{\pha{v}_k} \Bigr) \\
    &= \pha{\varpi}_0 + \eta {e^{j\varphi}} \Bigl( \phaconj{\varsigma}_k^{\ast} - \textstyle\sum\nolimits_{l = 1}^N {\pha{{y}}_{kl}} \frac{\pha{\phi}_{l1}}{\pha{\phi}_{k1}} \Bigr) = {\pha{\lambda}_1},
\end{aligned}
\end{equation}
where the last equality holds due to $\phamat{A} {\phavec{\phi}_1} = {\pha{\lambda}_1} {\phavec{\phi}_1}$, namely, $ \pha{\varpi}_0\pha{\phi}_{k1} + \eta {e^{j\varphi}} \bigl(\phaconj{\varsigma}_k^{\ast} \pha{\phi}_{k1} - \sum\nolimits_{m = 1}^N \pha{{y}}_{km} \pha{\phi}_{m1}\bigr) = {\pha{\lambda}_1} \pha{\phi}_{k1}$. It follows from \eqref{eq-proof-sync} that the complex frequencies converge to the dominant eigenvalue, i.e., ${\pha{\varpi}_{\rm sync}} = {\pha{\lambda}_1}$. Power system operation requires that only one fundamental-frequency component is present in the synchronous state. This is an additional requirement for Definition~\ref{def-complex-freq-sync-def}. It is apparent that $\re{\pha{\lambda}_2} < 0$ fulfills this requirement since the non-dominant modes decay to zero. 

We finally consider the special case where ${\phavec{\psi}_1\trans} {\phavec{v}_0} = 0$. In this case, the dominant modal response ${\phavec{v}^1} = {\phavec{\phi}_1}{\phavec{\psi}_1\trans}{\phavec{v}_0}{e^{{\pha{\lambda}_1}t}} = \mathbbb{0}_N$. The voltages converge to $\mathbbb{0}_N$ because of $\re{\pha{\lambda}_2} < 0$.
\endproof

\proof[Proof of Throrem~\ref{theorem-suff}] 
We first show that Condition~\ref{condition-suff} guarantees that the voltages with initial condition satisfying ${\phavec{\psi}_1\trans}{\phavec{v}_0} \neq 0$ converge to the eigenspace spanned by ${\phavec{\phi}_1}$, which is denoted by $\mathcal{S}$. A matrix $\phamat{P} \coloneqq \mat{I}_N - {\phavec{\phi}_1} {\phavec{\phi}_1\hermconj}/({\phavec{\phi}_1\hermconj} {\phavec{\phi}_1})$ denotes the projector onto the subspace orthogonal to $\mathcal{S}$. Then, the distance of the point $\phavec{v}$ to the set $\mathcal{S}$ is denoted as $\norm{\phavec{v}} _{\mathcal{S}} = \norm{\phamat{P}\,\phavec{v}}$. Notice that $\phamat{P}\hermconj = \phamat{P}$ and $\phamat{P}^2 = \phamat{P}$. The square of the distance is $\norm{\phavec{v}} _{\mathcal{S}} ^2 = \phavec{v}\hermconj \phamat{P}\, \phavec{v}$. We consider the Lyapunov function
\begin{equation}
\label{eq-lyap-fcn}
    V(\phavec{v}) \coloneqq \tfrac{1}{2} \norm{\phavec{v}} _{\mathcal{S}} ^2 = \tfrac{1}{2} \phavec{v}\hermconj \phamat{P}\, \phavec{v}.
\end{equation}
The time derivation of $V(\phavec{v})$ along the trajectories of the system \eqref{eq-fast-system} is given by $\dot{V}(\phavec{v}) = \frac{1}{2} \phavec{v}\hermconj (\phamat{A}\hermconj \phamat{P} + \phamat{P}\, \phamat{A}) \phavec{v}$ \footnote{The differential of $V(\phavec{v})$ is given by so-called Wirtinger derivatives as ${\rm d} V = \bigl(\frac{\partial V}{\partial \phavec{v}} \bigr) \trans {\rm d} \phavec{v} + \bigl (\frac{\partial V}{\partial \phavecconj{v}} \bigr) \trans {\rm d} \phavecconj{v}$ \cite[Th. 5.0.1]{hunger2007introduction}, where $\phavec{v}$ and $\phavecconj{v}$ can be treated as independent variables. It follows that $\dot V = \frac{1}{2} \phavec{v}\hermconj \phamat{P}\, \phamat{A}\, \phavec{v} + \frac{1}{2} \phavec{v}\trans \phamat{P}\trans\, \overline{\phamat{A}}\, \phavecconj{v} = \frac{1}{2} \phavec{v}\hermconj \phamat{P}\, \phamat{A}\, \phavec{v} + \frac{1}{2} (\phavec{v}\trans \phamat{P}\trans\, \overline{\phamat{A}}\, \phavecconj{v})\trans = \frac{1}{2} \phavec{v}\hermconj (\phamat{A}\hermconj \phamat{P} + \phamat{P}\, \phamat{A}) \phavec{v}$.}. From $\phamat{A} {\phavec{\phi}_1} = {\pha{\lambda}_1} {\phavec{\phi}_1}$ and $\phamat{P} = \mat{I}_N - {\phavec{\phi}_1} {\phavec{\phi}_1\hermconj}/({\phavec{\phi}_1\hermconj} {\phavec{\phi}_1})$, we obtain that $\phamat{A} - {\pha{\lambda}_1} \mat{I}_N = (\phamat{A} - {\pha{\lambda}_1}\mat{I}_N)\phamat{P}$, and further that $\phamat{A} = \phamat{A}\, \phamat{P} + {\pha{\lambda}_1} (\mat{I}_N - \phamat{P})$. It follows that $\phamat{P}\,\phamat{A} = \phamat{P}\,\phamat{A}\,\phamat{P}$ and $\phamat{A}\hermconj \phamat{P} = \phamat{P}\,\phamat{A}\hermconj \phamat{P}$. Therefore, we obtain 
\begin{equation}
\label{eq-v-dot}
\begin{aligned}
    \dot{V}(\phavec{v}) &= \tfrac{1}{2}\phavec{v}\hermconj \phamat{P} \bigl(\phamat{A} + \phamat{A}\hermconj\bigr) \phamat{P}\, \phavec{v} \\
    &= \eta \phavec{v}\hermconj \phamat{P} \left [\re{e^{j\varphi}{\rm diag} \left(\phavecconj{\varsigma}^{\ast}\right)} - \re{{e^{j\varphi}} \phamat{Y}} \right] \phamat{P}\, \phavec{v}.
\end{aligned}
\end{equation}
It follows that $\dot{V}(\phavec{v}) \leq 0$ holds if and only if
\begin{equation}
\label{eq-iff}
    \phavec{v}\hermconj \phamat{P} \re{e^{j\varphi}{\rm diag} \left(\phavecconj{\varsigma}^{\ast}\right)} \phamat{P}\, \phavec{v} \leq \phavec{v}\hermconj \phamat{P} \re{{e^{j\varphi}} \phamat{Y}} \phamat{P}\, \phavec{v}.
\end{equation}

Since $\re{e^{j\varphi}{\rm diag} \left(\phavecconj{\varsigma}^{\ast}\right)}$ is a diagonal matrix, the left-hand side of \eqref{eq-iff} is bounded as
\begin{equation}
\label{eq-left-hand-side}
    \phavec{v}\hermconj \phamat{P} \re{e^{j\varphi}{\rm diag} \left(\phavecconj{\varsigma}^{\ast}\right)} \phamat{P}\, \phavec{v} \leq \norm{\phavec{v}} _{\mathcal{S}} ^2 \max\limits_k \re{e^{j\varphi}{\phaconj{\varsigma}_k^{\ast}}}.
\end{equation}
To bound the right-hand side of \eqref{eq-iff}, we define another projector onto the orthogonal subspace of the span of $\mathbbm{1}_N$ as ${\mat{\Pi}} \coloneqq {\mat{I}_N} - \mathbbm{1}_N{\mathbbm{1}_N\trans}/N$ satisfying ${\mat{\Pi}\trans} = {\mat{\Pi}}$ and ${\mat{\Pi}^2} = {\mat{\Pi}}$. Since $\re{{e^{j\varphi}} \phamat{Y}}$ is a Laplacian matrix, its smallest eigenvalue is zero, and the corresponding right and left eigenvectors are $\mathbbm{1}_N$. It follows that $\re{{e^{j\varphi}} \phamat{Y}} = \re{{e^{j\varphi}} \phamat{Y}} \mat{\Pi} = \mat{\Pi} \re{{e^{j\varphi}} \phamat{Y}}$. We bound the right-hand side term of \eqref{eq-iff} as follows
\begin{equation*}
\begin{aligned}
    &\phavec{v}\hermconj \phamat{P} \re{{e^{j\varphi}} \phamat{Y}} \phamat{P}\, \phavec{v}
    = {{\phavec{v}}\hermconj}\phamat{P}{\mat{\Pi}}\re{{e^{j\varphi}} \phamat{Y}}{\mat{\Pi}}\phamat{P}\, \phavec{v} \\
    &\geq {\lambda _2}\bigl(\re{{e^{j\varphi}} \phamat{Y}}\bigr){\norm{{{\mat{\Pi}}\phamat{P}\, \phavec{v}}}^2}\\
    &= {\lambda _2}\bigl(\re{{e^{j\varphi}} \phamat{Y}}\bigr) \bigl( {{{\phavec{v}}\hermconj}\phamat{P}\, \phavec{v} - {{\phavec{v}}\hermconj}\phamat{P}\mathbbm{1}_N{\mathbbm{1}_N\trans}\phamat{P}\, \phavec{v}}/{N} \bigr)\\
    &= {\lambda _2}\bigl(\re{{e^{j\varphi}} \phamat{Y}}\bigr) \bigl( {{{\phavec{v}}\hermconj}\phamat{P}\, \phavec{v} - {{\phavec{v}}\hermconj}\phamat{P}(\phamat{P}\mathbbm{1}_N){{(\phamat{P}\mathbbm{1}_N)}\hermconj}\phamat{P}\, \phavec{v}}/{N} \bigr)\\
    &\geq {\lambda _2}\bigl(\re{{e^{j\varphi}} \phamat{Y}}\bigr)\norm{{\phavec{v}}}_{\mathcal{S}}^2 \bigl( {1 - {{\norm{{\phamat{P}\mathbbm{1}_N}}}^2}/{N}} \bigr)\\
    &= {\lambda _2}\bigl(\re{{e^{j\varphi}} \phamat{Y}}\bigr)\norm{{\phavec{v}}}_{\mathcal{S}}^2 \bigl( {1 - {{\mathbbm{1}_N\trans}\phamat{P}\mathbbm{1}_N}/{N}} \bigr)\\
    &= {\lambda _2}\bigl(\re{{e^{j\varphi}} \phamat{Y}}\bigr)\norm{{\phavec{v}}}_{\mathcal{S}}^2 \bigl( {{{\mathbbm{1}_N\trans} {\phavec{\phi}_1} {\phavec{\phi}_1\hermconj}\mathbbm{1}_N} / ({N{\phavec{\phi}_1\hermconj} {\phavec{\phi}_1}})} \bigr)\\
    &= {\lambda _2}\bigl(\re{{e^{j\varphi}} \phamat{Y}}\bigr) \left\| {\pha{v}} \right\|_{\mathcal{S}}^2 \frac{1}{N \sum\nolimits_{k=1}^N {{{\bigl\lvert {{{\pha{\phi}}_{k1}}} \bigr\rvert}^2}}} \times \\
    & \Bigl[ {\textstyle\sum\limits_{k=1}^N {{\bigl\lvert{\pha{\phi}}_{k1}\bigr\rvert}^2} + 2\textstyle\sum\limits_{k = 1}^N {\textstyle\sum\limits_{l = 1}^{l < k} {\bigl\lvert{\pha{\phi}}_{k1}\bigr\rvert \bigl\lvert {\pha{\phi}}_{l1}\bigr\rvert \cos \bigl(\angle{{\pha{\phi}}_{k1}} - \angle {{\pha{\phi}}_{l1}} \bigr )}}} \Bigr],
\end{aligned}
\end{equation*}
where the first inequality holds because the smallest Rayleigh quotient of ${\mat{\Pi}}\phamat{P}\,\phavec{v}$ orthogonal to the eigenvector $\mathbbm{1}_N$ of the smallest eigenvalue is the second smallest eigenvalue. The second inequality holds due to the Cauchy-Schwarz inequality.

When ${\phavec{\psi}_1\trans}{\phavec{v}_0} \neq 0$, it holds that $\lim_{t \to \infty} ({{\pha{v}}_k}/{{\pha{v}}_l}) = {{\pha{\phi}}_{k1}}/{{\pha{\phi}}_{l1}}$ as in \eqref{eq-limit-vkl}, where ${\pha{\phi}}_{l1} \neq 0,\ \forall l$ because Condition~\ref{condition-suff} requires the voltage amplitude ratio in the synchronous state to be nonzero and bounded. Then, the synchronous-state constraints in Condition~\ref{condition-suff} can be converted to $\bigl\lvert \angle{{\pha{\phi}}_{k1}} - \angle {{\pha{\phi}}_{l1}} \bigr\rvert \leq \bar{\delta}$ and $\bigl\lvert {| {\pha{\phi}}_{k1} |} / {| {\pha{\phi}}_{k1} |} - 1 \bigr\rvert \le \bar{\gamma} $. Following an analogous argument as in the proof of \cite[Lemma 2]{Gross-dVOC}, the term $\phavec{v}\hermconj \phamat{P} \re{{e^{j\varphi}} \phamat{Y}} \phamat{P}\, \phavec{v}$ can be bounded as 
\begin{equation}
\label{eq-right-hand-side}
\begin{aligned}
    \phavec{v}\hermconj \phamat{P} \re{{e^{j\varphi}} \phamat{Y}} \phamat{P}\, \phavec{v} &\geq {\lambda _2}\bigl(\re{{e^{j\varphi}} \phamat{Y}}\bigr)\norm {\pha{v}}_{\mathcal{S}}^2\tfrac{{\min\limits_k {{\lvert {{{\phavec{\phi}}_{k1}}} \rvert}^2}}}{{\max\limits_k {{\lvert {{{\phavec{\phi}}_{k1}}} \rvert}^2}}}\tfrac{{1 + \cos \bar{\delta}}}{2}\\
    &\geq {\lambda _2}\bigl(\re{{e^{j\varphi}} \phamat{Y}}\bigr)\norm {\pha{v}}_{\mathcal{S}}^2{\left(1 - \bar{\gamma} \right)^2}\tfrac{{1 + \cos \bar \delta}}{2}.
\end{aligned}
\end{equation}

By substituting \eqref{eq-left-hand-side} and \eqref{eq-right-hand-side} into \eqref{eq-v-dot}, we observe that the parametric relationship in \eqref{eq-suff-formula} is sufficient for
\begin{equation}
\label{eq-v-dot-less-than-0}
    \dot{V}(\phavec{v}) < 0, \quad \forall \phavec{v} \notin \mathcal{S}.
\end{equation}
We conclude from \eqref{eq-lyap-fcn} and \eqref{eq-v-dot-less-than-0} that the system \eqref{eq-fast-system} is asymptotically stable with respect to the eigenspace $\mathcal{S}$ for any initial states ${\phavec{v}_0}$ satisfying ${\phavec{\psi}_1\trans}{\phavec{v}_0} \neq 0$. From $\lim_{t \to \infty} ({{\pha{v}}_k}/{{\pha{v}}_l}) = {{\pha{\phi}}_{k1}}/{{\pha{\phi}}_{l1}}$, we further obtain that $\lim\limits_{t \to \infty} ({{\dot {\pha{v}}}_k}/{{\pha{v}}_k}) = {\pha{\lambda}_1}$ as in \eqref{eq-proof-sync}. This shows that the system \eqref{eq-fast-system} synchronizes at $\pha{\lambda}_1$.

Moreover, the asymptotic stability implies $\lim_{t \to \infty}\norm{\phavec{v}} _{\mathcal{S}} ^2 = 0$, which suffices to show $\re{\pha{\lambda}_k} < 0$ for all $k \geq 2$. Namely, Condition~\ref{condition-suff} $\Rightarrow$ Condition~\ref{condition-suff-nece}. It then follows that the voltages converge to $\mathbbb{0}_N$ for ${\phavec{\psi}_1\trans} {\phavec{v}_0} = 0$.
\endproof

\begin{lemma}
\label{lemma-slow-system-equili}
The system in \eqref{eq-slow-system-coa} has a unique equilibrium $[\vect{u}_s\trans, \vect{\delta} _s\trans, \vect{u}_s\trans]\trans$, where $[\vect{u}_s\trans, \vect{\delta} _s\trans]\trans$ is the solution of
\begin{equation}
\label{eq-slow-system-equ}
    \underbrace{\begin{bmatrix}
    \mat{G}^{\prime} + \alpha \mat{I}_N & -\mat{B}^{\prime} \\ \mat{B}^{\prime} & \mat{G}^{\prime} \\ \mathbbb{0}_N\trans & \mathbbm{1}_N\trans
    \end{bmatrix}}_{\mat{A}}
    \begin{bmatrix}
    \vect{u}_s \\ \vect{\delta} _s
    \end{bmatrix} = 
    \underbrace{\begin{bmatrix}
    \vect{\sigma}^{\prime\star} + \alpha \vect{u}^{\star} \\
    \vect{\rho}^{\prime\star} - \mathbbm{1}_N\mathbbm{1}_N\trans\vect{\rho}^{\prime\star}/N \\ 0
    \end{bmatrix}}_{\vect{b}}.
\end{equation}
\end{lemma}

\begin{proof}
    We assume that there are equilibria in the system \eqref{eq-slow-system-coa}. Then, it holds that $\dot {\vect{u}} = \mathbbb{0}_N$, $\dot {\vect{\delta}} = \mathbbb{0}_N$, and $\vect{u} = \vect{u}_f$ in the equilibria. Thus, \eqref{eq-slow-system-equ} follows from \eqref{eq-slow-system-coa}. The set of equations \eqref{eq-slow-system-equ} has a unique solution if and only if the rank condition ${\rm rank} \left( \mat{A} \right) = {\rm rank} \left(\left [\mat{A}, \vect{b} \right] \right) = 2N$ holds. The top $N$ rows of $\mat{A}$ are linearly independent because of the presence of $\alpha \mat{I}_N$. The middle $N$ rows have a row rank of $N - 1$ since $\mat{B}^{\prime}$ and $\mat{G}^{\prime}$ are Laplacian. Furthermore, the middle $N$ rows of $\left[\mat{A}, \vect{b}\right]$ have a row rank of $N - 1$ due to $\mathbbm{1}_N\trans \left[\mat{B}^{\prime}, \mat{G}^{\prime}, \vect{\rho}^{\prime\star} - \mathbbm{1}_N \mathbbm{1}_N \trans \vect{\rho} ^{\prime\star}/N \right] = \left[ \mathbbb{0}_{N} \trans, \mathbbb{0}_{N}\trans, 0\right]$. Therefore, the rank condition holds, indicating the system \eqref{eq-slow-system-coa} has a unique equilibrium.
\end{proof}

\proof[Proof of Throrem~\ref{theorem-volt-stab-guarant}] Conciser the error coordinates as
$[\tilde{\vect{u}}\trans, \tilde{\vect{\delta}}\trans, \tilde{\vect{u}}_f\trans]\trans \coloneqq [\vect{u}\trans, \vect{\delta} \trans, \vect{u}_f\trans]\trans - [\vect{u}_s\trans, \vect{\delta} _s\trans, \vect{u}_s\trans]\trans$. The error dynamics are then given as
\begin{equation}
\label{eq-slow-system-to-origin}
\begin{aligned}
    \dot {\tilde{\vect{u}}} &= - \eta \mat{G}^{\prime} \tilde{\vect{u}} + \eta \mat{B}^{\prime} \tilde{\vect{\delta}} - \eta \alpha \tilde{\vect{u}}_f \\
    \dot {\tilde{\vect{\delta}}} &= - \eta \mat{B}^{\prime} \tilde{\vect{u}} - \eta \mat{G}^{\prime} \tilde{\vect{\delta}}\\
    \tau \dot{\tilde{\vect{u}}}_f &= \tilde{\vect{u}} - \tilde{\vect{u}}_f.
\end{aligned}
\end{equation}
A positive-definite Lyapunov function is defined as
\begin{equation*}
    W({\tilde{\vect{u}}},{\tilde{\vect{\delta}}_f},\tilde{\vect{u}}) = \tfrac{1}{2}{{\tilde{\vect{u}}}\trans} \tilde{\vect{u}} + \tfrac{1}{2}{\tilde{\vect{\delta}}\trans} \tilde{\vect{\delta}} + \tfrac{1}{2}\eta \alpha \tau {\tilde{\vect{u}}_f \trans} \tilde{\vect{u}}_f.
\end{equation*}
The time derivative of it along \eqref{eq-slow-system-to-origin} is 
\begin{equation*}
    \dot {W}({\tilde{\vect{u}}},{\tilde{\vect{\delta}}_f},\tilde{\vect{u}}) = -\eta {\tilde{\vect{u}}\trans} \mat{G}^{\prime} \tilde{\vect{u}} - \eta {\tilde{\vect{\delta}}\trans} \mat{G}^{\prime} \tilde{\vect{\delta}} -\eta \alpha \tilde{\vect{u}}_f \trans \tilde{\vect{u}}_f,
\end{equation*}
which is negative semi-definite. From $\dot {W}({\tilde{\vect{u}}},{\tilde{\vect{\delta}}_f},\tilde{\vect{u}}) = 0$, \eqref{eq-slow-system-to-origin}, and $\mathbbm{1}_N\trans \tilde{\vect{\delta}} = 0$, we obtain that $[\tilde{\vect{u}}\trans, \tilde{\vect{\delta}}\trans, \tilde{\vect{u}}_f\trans]\trans = \mathbbb{0}_{3N}$. It follows from the LaSalle's invariance principle \cite[Corollary 4.2]{Khalil-nonlinear} that the system \eqref{eq-slow-system-to-origin} is globally asymptotically stable with respect to the origin. Equivalently, the system \eqref{eq-slow-system-coa} is globally asymptotically stable with respect to the equilibrium $[\vect{u}_s\trans, \vect{\delta} _s\trans, \vect{u}_s\trans]\trans$. For the system \eqref{eq-slow-system-real-coeff}, ${\vect{u}} \to \vect{u}_s$ and $\dot {\vect{\theta}}$ (i.e., $\dot {\vect{\delta}} + \mathbbm{1}_N \dot{\theta}_0$) $\to \mathbbm{1}_N \dot{\theta}_0$ as $t \to \infty$. The steady-state frequency is $ \dot{\theta}_0 = {\omega _0} + \eta \mathbbm{1}_N\trans \vect{\rho}^{\prime\star}/N = {\omega _0} + \eta \mathbbm{1}_N\trans \im{{e^{j\varphi}} \phavecconj{\varsigma}^{\star}}/N$.
\endproof

\ifCLASSOPTIONcaptionsoff
  \newpage
\fi


\bibliographystyle{IEEEtran}
\bibliography{IEEEabrv,Bibliography}

\begin{thebibliography}{10}
\providecommand{\url}[1]{#1}
\csname url@rmstyle\endcsname
\providecommand{\newblock}{\relax}
\providecommand{\bibinfo}[2]{#2}
\providecommand\BIBentrySTDinterwordspacing{\spaceskip=0pt\relax}
\providecommand\BIBentryALTinterwordstretchfactor{4}
\providecommand\BIBentryALTinterwordspacing{\spaceskip=\fontdimen2\font plus
\BIBentryALTinterwordstretchfactor\fontdimen3\font minus \fontdimen4\font\relax}
\providecommand\BIBforeignlanguage[2]{{%
\expandafter\ifx\csname l@#1\endcsname\relax
\typeout{** WARNING: IEEEtran.bst: No hyphenation pattern has been}%
\typeout{** loaded for the language `#1'. Using the pattern for}%
\typeout{** the default language instead.}%
\else
\language=\csname l@#1\endcsname
\fi
#2}}
\renewcommand\BIBentryALTinterwordstretchfactor{4}

\bibitem{Synchronization-meaning}
A.~Sajadi, R.~W. Kenyon, and B.-M. Hodge, ``Synchronization in electric power networks with inherent heterogeneity up to 100\% inverter-based renewable generation,'' \emph{Nat. Commun.}, vol.~13, no.~1, pp. 1--12, 2022.

\bibitem{Kundur-def}
P.~Kundur, J.~Paserba, V.~Ajjarapu, G.~Andersson, A.~Bose, C.~Canizares, N.~Hatziargyriou, D.~Hill, A.~Stankovic, C.~Taylor, T.~Van~Cutsem, and V.~Vittal, ``Definition and classification of power system stability {IEEE/CIGRE} joint task force on stability terms and definitions,'' \emph{{IEEE} Trans. Power Syst.}, vol.~19, no.~3, pp. 1387--1401, 2004.

\bibitem{Bergen-models}
A.~Bergen and D.~Hill, ``A structure preserving model for power system stability analysis,'' \emph{{IEEE} Trans. Power Appar. Syst.}, vol. PAS-100, no.~1, pp. 25--35, 1981.

\bibitem{Kundur-book}
P.~Kundur, N.~J. Balu, and M.~G. Lauby, \emph{Power System Stability and Control}.\hskip 1em plus 0.5em minus 0.4em\relax New York, NY, USA: McGraw-Hill, 1994.

\bibitem{NERC-TR-2017}
{NERC/WECC Joint Task Force}, ``1200 {MW} fault induced solar photovoltaic resource interruption disturbance report,'' {NERC}, Atlanta, GA, USA, Tech. Rep., 2017.

\bibitem{De-coupling}
K.~De~Brabandere, B.~Bolsens, J.~Van~den Keybus, A.~Woyte, J.~Driesen, and R.~Belmans, ``A voltage and frequency droop control method for parallel inverters,'' \emph{{IEEE} Trans. Power Electron.}, vol.~22, no.~4, pp. 1107--1115, 2007.

\bibitem{He-PLL}
X.~He, H.~Geng, R.~Li, and B.~C. Pal, ``Transient stability analysis and enhancement of renewable energy conversion system during {LVRT},'' \emph{{IEEE} Trans. Sustain. Energy}, vol.~11, no.~3, pp. 1612--1623, 2020.

\bibitem{Huang-droop}
L.~Huang, H.~Xin, Z.~Wang, L.~Zhang, K.~Wu, and J.~Hu, ``Transient stability analysis and control design of droop-controlled voltage source converters considering current limitation,'' \emph{{IEEE} Trans. Smart Grid}, vol.~10, no.~1, pp. 578--591, 2019.

\bibitem{Ge-VSG}
P.~Ge, C.~Tu, F.~Xiao, Q.~Guo, and J.~Gao, ``Design-oriented analysis and transient stability enhancement control for a virtual synchronous generator,'' \emph{{IEEE} Trans. Ind. Electron.}, vol.~70, no.~3, pp. 2675--2684, 2023.

\bibitem{Arghir-matching}
C.~Arghir, T.~Jouini, and F.~Dörfler, ``Grid-forming control for power converters based on matching of synchronous machines,'' \emph{Automatica}, vol.~95, pp. 273--282, 2018.

\bibitem{Colombino-dVOC}
M.~Colombino, D.~Groß, J.-S. Brouillon, and F.~Dörfler, ``Global phase and magnitude synchronization of coupled oscillators with application to the control of grid-forming power inverters,'' \emph{{IEEE} Trans. Autom. Control}, vol.~64, no.~11, pp. 4496--4511, 2019.

\bibitem{Gross-dVOC}
D.~Groß, M.~Colombino, J.-S. Brouillon, and F.~Dörfler, ``The effect of transmission-line dynamics on grid-forming dispatchable virtual oscillator control,'' \emph{{IEEE} Trans. Control Netw. Syst.}, vol.~6, no.~3, pp. 1148--1160, 2019.

\bibitem{Subotic-dVOC}
I.~Subotić, D.~Groß, M.~Colombino, and F.~Dörfler, ``A {Lyapunov} framework for nested dynamical systems on multiple time scales with application to converter-based power systems,'' \emph{{IEEE} Trans. Autom. Control}, vol.~66, no.~12, pp. 5909--5924, 2021.

\bibitem{Lu-benchmarking}
M.~Lu, S.~Dhople, and B.~Johnson, ``Benchmarking nonlinear oscillators for grid-forming inverter control,'' \emph{{IEEE} Trans. Power Electron.}, vol.~37, no.~9, pp. 10\,250--10\,266, 2022.

\bibitem{Milano-complex-freq}
F.~Milano, ``Complex frequency,'' \emph{{IEEE} Trans. Power Syst.}, vol.~37, no.~2, pp. 1230--1240, 2022.

\bibitem{Gu-complex-angle}
Y.~Li, T.~C. Green, and Y.~Gu, ``The intrinsic communication in power systems: A new perspective to understand synchronization stability,'' \emph{{IEEE} Trans. Circuits Syst. I-Regul. Pap.}, vol.~70, no.~11, pp. 4615--4626, 2023.

\bibitem{Zhong-complex-freq}
W.~Zhong, G.~Tzounas, M.~Liu, and F.~Milano, ``On-line inertia estimation of virtual power plants,'' \emph{Electr. Power Syst. Res.}, vol. 212, p. 108336, 2022.

\bibitem{Sanniti-curvature}
F.~Sanniti, G.~Tzounas, R.~Benato, and F.~Milano, ``Curvature-based control for low-inertia systems,'' \emph{{IEEE} Trans. Power Syst.}, vol.~37, no.~5, pp. 4149--4152, 2022.

\bibitem{moutevelis2022taxonomy}
D.~Moutevelis, J.~Roldán-Pérez, M.~Prodanovic, and F.~Milano, ``Taxonomy of power converter control schemes based on the complex frequency concept,'' \emph{{IEEE} Trans. Power Syst.}, vol.~39, no.~1, pp. 1996--2009, 2024.

\bibitem{conteville2012linear}
L.~Conteville and E.~Panteley, ``Linear reformulation of the {Kuramoto} model: Asymptotic mapping and stability properties,'' in \emph{CCCA12}.\hskip 1em plus 0.5em minus 0.4em\relax IEEE, 2012, pp. 1--6.

\bibitem{panteley2020practical}
E.~Panteley, A.~Lor{\'\i}a, and A.~El-Ati, ``Practical dynamic consensus of {Stuart-Landau} oscillators over heterogeneous networks,'' \emph{International Journal of Control}, vol.~93, no.~2, pp. 261--273, 2020.

\bibitem{rodrigues2016kuramoto}
F.~A. Rodrigues, T.~K.~D. Peron, P.~Ji, and J.~Kurths, ``The {Kuramoto} model in complex networks,'' \emph{Physics Reports}, vol. 610, pp. 1--98, 2016.

\bibitem{Milano-interpretation}
F.~Milano, ``A geometrical interpretation of frequency,'' \emph{{IEEE} Trans. Power Syst.}, vol.~37, no.~1, pp. 816--819, 2021.

\bibitem{Ilic-lnv}
M.~Ilic, ``Network theoretic conditions for existence and uniqueness of steady state solutions to electric power circuits,'' in \emph{Proc. IEEE Int. Symp. Circuits Syst.}, vol.~6, 1992, pp. 2821--2828.

\bibitem{Chen-dv-v}
L.~Chen, Y.~Min, and W.~Hu, ``An energy-based method for location of power system oscillation source,'' \emph{{IEEE} Trans. Power Syst.}, vol.~28, no.~2, pp. 828--836, 2012.

\bibitem{Kogler-norm}
R.~Kogler, A.~Plietzsch, P.~Schultz, and F.~Hellmann, ``Normal form for grid-forming power grid actors,'' \emph{PRX Energy}, vol.~1, no.~1, p. 013008, 2022.

\bibitem{Yang-augmented-sync}
P.~Yang, F.~Liu, T.~Liu, and D.~J. Hill, ``Augmented synchronization of power systems,'' \emph{{IEEE} Trans. Autom. Control}, pp. 1--16, 2023.

\bibitem{Zhao-dc-microgrid}
J.~Zhao and F.~D{\"o}rfler, ``Distributed control and optimization in dc microgrids,'' \emph{Automatica}, vol.~61, pp. 18--26, 2015.

\bibitem{Stott-dc-power}
B.~Stott, J.~Jardim, and O.~Alsac, ``{DC} power flow revisited,'' \emph{{IEEE} Trans. Power Syst.}, vol.~24, no.~3, pp. 1290--1300, 2009.

\bibitem{European-code}
E.~Commission, ``{Commission Regulation (EU)} 2016/631 of 14 {April} 2016, establishing a network code on requirements for grid connection of generators,'' \emph{Off. J. Eur. Union}, 2016.

\bibitem{Khalil-nonlinear}
H.~K. Khalil, \emph{Nonlinear Systems}, 3rd~ed.\hskip 1em plus 0.5em minus 0.4em\relax Englewood Cliffs, NJ, USA: Prentice-Hall, 2002.

\bibitem{Xin-SISO-equ}
L.~Xu, H.~Xin, L.~Huang, H.~Yuan, P.~Ju, and D.~Wu, ``Symmetric admittance modeling for stability analysis of grid-connected converters,'' \emph{{IEEE} Trans. Energy Convers.}, vol.~35, no.~1, pp. 434--444, 2019.

\bibitem{menara2020cluster}
T.~Menara, G.~Baggio, D.~S. Bassett, and F.~Pasqualetti, ``Stability conditions for cluster synchronization in networks of heterogeneous {Kuramoto} oscillators,'' \emph{{IEEE} Trans. Control Netw. Syst.}, vol.~7, no.~1, pp. 302--314, 2020.

\bibitem{hunger2007introduction}
\BIBentryALTinterwordspacing
R.~Hunger, ``An introduction to complex differentials and complex differentiability,'' 2007. [Online]. Available: \url{https://mediatum.ub.tum.de/doc/631019/document.pdf}
\BIBentrySTDinterwordspacing

\end{thebibliography}

\begin{IEEEbiography}[{\includegraphics[width=1in,height=1.25in,clip,keepaspectratio]{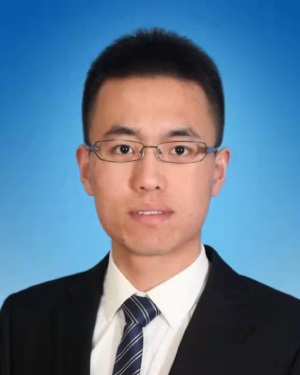}}]{Xiuqiang He}
received his B.S. degree and Ph.D. degree in control science and engineering from Tsinghua University, China, in 2016 and 2021, respectively. Since 2021, he has been a Postdoctoral Researcher with the Automatic Control Laboratory, ETH Zürich, Switzerland. His current research interests include power system dynamics, stability, and control, involving multidisciplinary expertise in automatic control, power systems, power electronics, and renewable energy sources. Dr. He was the recipient of the Beijing Outstanding Graduates Award and the Outstanding Doctoral Dissertation Award from Tsinghua University.
\end{IEEEbiography}

\begin{IEEEbiography}[{\includegraphics[width=1in,height=1.25in,clip,keepaspectratio]{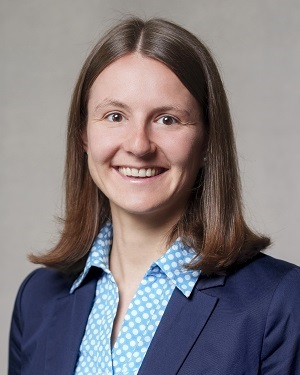}}]{Verena Häberle}
received the B.Sc. and M.Sc. degrees in electrical engineering and information technology from ETH Zürich, Switzerland, in 2018 and 2020, respectively, where she is currently pursuing the Ph.D. degree with the Automatic Control Laboratory. For her outstanding academic achievements during her Master’s thesis with the Automatic Control Laboratory, ETH Zürich, under Prof. F. Dörfler, she was honored with ETH Medal and the SGA Award from the Swiss Society of Automatic Control. Her research focuses on the control design of dynamic virtual power plants in future power systems.
\end{IEEEbiography}

\begin{IEEEbiography}[{\includegraphics[width=1in,height=1.25in,clip,keepaspectratio]{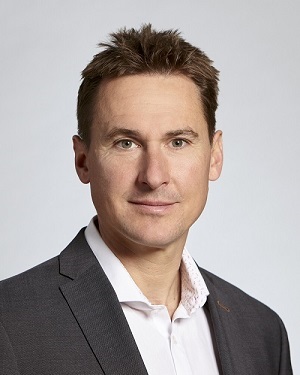}}]{Florian Dörfler}
is a Professor at the Automatic Control Laboratory at ETH Zürich. He received his Ph.D. degree in Mechanical Engineering from the University of California at Santa Barbara in 2013, and a Diplom degree in Engineering Cybernetics from the University of Stuttgart in 2008. From 2013 to 2014 he was an Assistant Professor at the University of California Los Angeles. His research interests are centered around automatic control, system theory, and optimization. His particular foci are on network systems, data-driven settings, and applications to power systems. He is currently serving on the council of the European Control Association and as a senior editor of Automatica.
\end{IEEEbiography}

\end{document}